\documentclass[reqno,11pt]{amsart}

\usepackage[latin1]{inputenc}
\usepackage{amsfonts}
\usepackage{amssymb}
\usepackage{graphicx}         
\usepackage{tikz-cd}        
\usepackage{tikz}
\usetikzlibrary{positioning}
\usepackage{caption}       
\usepackage{mathtools}

\newtheorem{thm}{Theorem}[section]
\newtheorem{cor}[thm]{Corollary}
\newtheorem{lemma}[thm]{Lemma}

\newtheorem{prop}[thm]{Proposition}
\theoremstyle{definition}
\newtheorem{defn}[thm]{Definition}
\theoremstyle{remark}

\newtheorem{example}[thm]{Example}

\renewcommand{\a}{\alpha}
\renewcommand{\b}{\beta}
\newcommand{\ga}{\gamma}
\newcommand{\Ga}[1]{\Gamma_{#1}}

\newcommand{\PB}{\left\{\cdot\,,\cdot\right\}}
\newcommand{\pB}[1]{\left\{#1,\cdot\right\}}
\newcommand{\Pb}[1]{\left\{\cdot\,,#1\right\}}
\newcommand{\pb}[1]{\left\{#1\right\}}

\renewcommand{\(}{\left(}
\renewcommand{\)}{\right)}

\newcommand{\set}[1]{\left\{#1\right\}}

\newcommand{\cL}{\mathcal L}

\newcommand{\cX}{\mathcal X}
\newcommand{\X}{ \cX}

\newcommand{\bbC}{\mathbb C}
\newcommand{\bbF}{\mathbb F}

\newcommand{\bbN}{\mathbb N}
\newcommand{\bbR}{\mathbb R}

\newcommand{\KP}{KP property}
\newcommand{\tP}{\tilde P}
\newcommand{\tx}[1]{\tilde x_#1}
\newcommand{\ty}[2]{\bar y_{#1,#2}}

\newcommand{\cycl}{\circlearrowleft}

\newcommand{\LV}{\mathop{\rm LV}\nolimits}

\newcommand{\leqs}{\leqslant}

\newcommand{\Rk}{\hbox{Rk\,}}

\newcommand{\wght}{\varpi}

\newcommand{\Gb}{\Delta}
\newcommand{\PBb}{\PB_b}
\newcommand{\Pbb}[1]{\Pb{#1}_b}
\newcommand{\pbb}[1]{\pb{#1}_b}

\newcommand{\Kb}{K_b}

\newif\ifprivate
\privatefalse

 \numberwithin{equation}{section}

\def\???{\ifprivate {\bf {???}} \marginpar{{\Huge {\bf ?}}}\else \fi}
\numberwithin{equation}{section}

\begin{document}

\nocite{*}

\title[Kahan discretization and Poisson maps]{Kahan discretizations of skew-symmetric Lotka-Volterra systems and
   Poisson maps}

\author[C. A. Evripidou]{C. A. Evripidou}
\address{Charalampos Evripidou, Department of Mathematics and Statistics, University of Cyprus, P.O.~Box 20537,
  1678 Nicosia, Cyprus,\newline
  Department of Mathematics and Statistics, La Trobe University, Melbourne, Victoria 3086, Australia}
\email{evripidou.charalambos@ucy.ac.cy}

\author[P. Kassotakis]{P. Kassotakis}
\address{Pavlos Kassotakis, Department of Mathematics and Statistics, University of Cyprus, P.O.~Box 20537, 1678
  Nicosia, Cyprus} \email{pavlos1978@gmail.com}

\author[P. Vanhaecke]{P. Vanhaecke}
\address{Pol Vanhaecke, Laboratoire de Math\'ematiques et Applications, UMR 7348 CNRS, Universit\'e de Poitiers, 11
  Boulevard Marie et Pierre Curie, Téléport 2 - BP 30179, 86 962 Chasseneuil Futuroscope Cedex,
  France}\email{pol.vanhaecke@math.univ-poitiers.fr}

\date{\today}


\begin{abstract}
The Kahan discretization of the Lotka-Volterra system, associated with any skew-symmetric graph $\Gamma$, leads to
a family of rational maps, parametrized by the step size. When these maps are Poisson maps with respect to the
quadratic Poisson structure of the Lotka-Volterra system, we say that the graph $\Gamma$ has the Kahan-Poisson
property. We show that if $\Gamma$ is connected, it has the Kahan-Poisson property if and only if it is a cloning
of a graph with vertices $1,2,\dots,n$, with an arc $i\to j$ precisely when $i<j$, and with all arcs having the
same value. We also prove a similar result for augmented graphs, which correspond with deformed Lotka-Volterra
systems and show that the obtained Lotka-Volterra systems and their Kahan discretizations are superintegrable as
well as Liouville integrable.
\end{abstract}

\maketitle

\setcounter{tocdepth}{1}

\tableofcontents

\section{Introduction}
With any complex skew-symmetric $n\times n$ matrix $A=(a_{i,j})$ is associated a (skew-symmetric) Lotka-Volterra
system, which is the Hamiltonian system, described in terms of the standard coordinates $x_1,\dots,x_n$ of $\bbC^n$
by the following system of differential equations:
\begin{equation}\label{eq:LV_gen_intro}
  \dot x_i = \sum_{j=1}^n a_{i,j} x_i x_j\;, \ \ i=1,2, \dots , n \; .
\end{equation}
The Hamiltonian structure, which is also determined by $A$, is defined by the basic Poisson brackets
$\pb{x_i,x_j}=a_{i,j}x_ix_j$, for $1\leqslant i,j\leqslant n$, with Hamiltonian $H:=x_1+x_2+\cdots+x_n$. The matrix
$A$ may be viewed as the adjacency matrix of a (skew-symmetric) graph $\Gamma$, having the integers $1,2,\dots,n$
as vertices and with an arc from $i$ to $j$ with value $a_{i,j}$ when $a_{i,j}\neq0$ and~$i<j$. We often think of
the Lotka-Volterra system as being associated with $\Gamma$ and denote it by $\LV(\Gamma)$. Notice that $\Gamma$ is
determined, up to isomorphism, by $\LV(\Gamma)$, as was shown in \cite{PPP_graphs}. When the entries of $A$ are all
real, one may also consider~\eqref{eq:LV_gen_intro} on $\bbR^n$, so in what follows we suppose that the entries of
$A$ belong to $\bbF$, where $\bbF$ stands for $\bbR$ or $\bbC$ and we consider \eqref{eq:LV_gen_intro} on $\bbF^n$.

For a system of quadratic differential equations, such as \eqref{eq:LV_gen_intro}, a natural discretization has
been constructed by Kahan \cite{kahan}, leading to a rational map, called its Kahan map (see also
\cite{HK,HK2}). Applied to \eqref{eq:LV_gen_intro}, the Kahan map with step size $\varepsilon\in\bbF^*$ is the
rational map, corresponding to the field automorphism~\ $\widetilde{}\;:\bbF(x_1,\dots,x_n)\to
\bbF(x_1,\dots,x_n)$, defined by the following formulas, where the right hand side has been obtained from
\eqref{eq:LV_gen_intro} by polarization:

\begin{equation*}
  \frac{\tx i -x_i}{\varepsilon} = \sum_{j=1}^n a_{i,j} \(\tx i x_j+x_i\tx j\)\;, \ \ i=1,2, \dots , n \; .
\end{equation*}

From the point of view of discrete integrability, a natural question is whether the Kahan map is a Poisson map with
respect to the above Poisson structure. In formulas, this means that $\pb{\tx i,\tx j}=a_{i,j}\tx i\tx j$, for
$1\leqslant i<j\leqslant n.$ When this is the case, we say that $\Gamma$ has the Kahan-Poisson property. A general
skew-symmetric graph $\Gamma$ does not have the Kahan-Poisson property. In order to give a more precise answer, let
us denote by $\Ga n$ the skew-symmetric graph with vertices $1,\dots,n$, and with an arc $i\to j$, with value $1$,
for any $i<j$; for $\gamma\in\bbF^*$ we denote by $\gamma\Ga n$ the graph with the same vertices and arcs as $\Ga
n$, but where the value of every arc is $\gamma$.
\begin{thm}\label{thm:main}
  Let $\Gamma$ be a connected skew-symmetric graph. Then $\Gamma$ has the Kahan-Poisson property if and only if
  $\Gamma$ is isomorphic to $\gamma\Ga n^\wght$ for some $n\in\bbN^*$, some $\gamma\in \bbF^*$ and some weight
  vector $\wght$ for $\Ga n$.
\end{thm}
For a skew-symmetric graph which is not connected, Theorem \ref{thm:main} applies to each one of its connected
components.

The notions of cloning and decloning of skew-symmetric graphs and Lotka-Volterra systems were introduced in
\cite{PPP_graphs} in the study of morphisms and automorphisms of graphs and Lotka-Volterra systems. For a graph,
decloning amounts to identifying two (or more) vertices when they have the same neighborhood (which means that the
corresponding lines of its adjacency matrix $A$ are identical); the quotient graph is then said to be
irreducible. With this terminology, the theorem can also be stated by saying that the only connected, irreducible
skew-symmetric graphs $\Gamma$ which have the Kahan-Poisson property are the graphs $\gamma\Ga n$, where
$n\in\bbN^*$ and $\gamma\in \bbF^*$.

The fact that $\Ga n$ (and hence $\gamma\Ga n$, for all $\gamma$) has the Kahan-Poisson property has already been
shown in \cite[Prop.\ 3.8]{KKQTV}. We show in Proposition~\ref{prp:kahan_cloning} that the Kahan-Poisson property
is preserved by cloning and decloning. This proves the easier, inverse implication in Theorem~\ref{thm:main}. The
main result which we prove in this paper is the direct implication, which we first show in dimension~3
(Section~\ref{sec:3dim}), then in dimension 4, using the result in dimension~3 (Section~\ref{par:dim4}), and
finally in dimension $n>4$, using the result in dimension~4 (Section~\ref{par:dimn}).

In Section \ref{sec:LV_deformed}, we prove the following generalization of Theorem~\ref{thm:main} to augmented
graphs, which correspond to deformed Lotka-Volterra systems (see Section \ref{sec:LV_deformed} for the definition
of an augmented graph and the \KP\ for such graphs):
\begin{thm}\label{thm:main_deformed}
  Let $\Delta$ be an augmented graph of a connected skew-symmetric graph $\Gamma$. Then $\Delta$ has the
  Kahan-Poisson property if and only if $\Delta$ is isomorphic to an augmented graph of $\ga\Ga n^\wght$ for some
  $n\in\bbN^*$, some $\gamma\in \bbF^*$ and some weight vector $\wght$ for $\Ga n$.
\end{thm}

We show in Section \ref{sec:integ} that the (deformed) Lotka-Volterra systems, corresponding to the (augmented)
graphs which appear in Theorems \ref{thm:main} and \ref{thm:main_deformed}, have a discretization with good
integrability properties, namely the Kahan discretization, which is in these cases a Poisson map with respect to
the original Poisson structure, is both superintegrable and Liouville integrable. It follows that the (deformed)
Lotka-Volterra systems whose Kahan discretization is integrable with respect to their original Poisson structure
are characterized by the Kahan-Poisson property.

\section{The Kahan-Poisson property}
In this section we introduce the Kahan-Poisson property for skew-symmetric graphs and establish it for a
particular family of such graphs. We first recall the basic facts which we will use about the Kahan discretization of
systems of quadratic differential equations and about skew-symmetric Lotka-Volterra systems; see
\cite{PPP_graphs,kahan} for details.

\subsection{The Kahan discretization} Consider a system of differential equations on $\bbF^n$:
\begin{equation}\label{eq:quadratic_system}
  \dot x_i = Q_i(x)\;, \ \ i=1,2, \dots , n \; .
\end{equation}
Here, $x=(x_1,\dots,x_n)$ and $Q_i$ is assumed to be a quadratic form, whose corresponding symmetric bilinear form
is denoted by $B_i$, so that $Q_i(x)=B_i(x,x)$. The \emph{Kahan discretization} of \eqref{eq:quadratic_system} is
by definition given by
\begin{equation}\label{eq:kahan_general}
  \frac{\tilde x_i-x_i}{2\varepsilon} = B_i(\tilde x,x)\;, \quad i=1,2, \dots , n \; ,
\end{equation}
where $\varepsilon\in\bbF^*$ is a non-zero parameter, the \emph{step size}.  When (linearly!) solved for $\tilde
x_1,\dots,\tilde x_n$, one gets a family of birational maps from $\bbF^n$ to itself, parametrized by
$\varepsilon\in\bbF^*$. Thinking of $\varepsilon$ as being fixed, it is called the \emph{Kahan map}. We will mostly
work with the corresponding endomorphism~$K$ of the field of rational functions $\bbF(x)=\bbF(x_1,x_2,\dots,x_n)$,
defined by $K(x_i):=\tilde x_i$, for $i=1,\dots,n$; we call it the \emph{Kahan morphism}.

\subsection{Lotka-Volterra systems}
We are interested in the Kahan discretization of skew-symmetric Lotka-Volterra systems. As recalled in the
introduction, a Lotka-Volterra system is associated with any skew-symmetric $n\times n$ matrix $A$; we also view
$A=(a_{i,j})$ as the adjacency matrix of a skew-symmetric graph $\Gamma=(S,A)$, with vertex set
$S=\set{1,2,\dots,n}$, and think of the Lotka-Volterra system as being associated with the graph $\Gamma$, denoted
$\LV(\Gamma)$. The Poisson structure of~$\LV(\Gamma)$, which we consider here as a Poisson bracket on $\bbF(x)$, is
the quadratic bracket, given by $\pb{x_i,x_j}=a_{i,j}x_ix_j$, for $i,j=1,2,\dots,n$. It makes $(\bbF(x),\PB)$ into
a Poisson field. The Hamiltonian vector field~$\X_H=\Pb{H}$ on $\bbF^n$, associated with the Hamiltonian
$H:=x_1+x_2+\dots+x_n$, is given by the following quadratic differential equations:
\begin{equation*}
  \dot x_i = \sum_{j=1}^n a_{i,j} x_i x_j\;, \ \ i=1,2, \dots , n \; .
\end{equation*}
Its Kahan map is defined by the following specialisation of~\eqref{eq:kahan_general}:
\begin{equation}\label{eq:LV_Kahan}
  \frac{\tx i -x_i}{\varepsilon} = \tx i\sum_{j=1}^n a_{i,j}x_j +x_i\sum_{j=1}^n a_{i,j}\tx j\;, \ \
  i=1,2, \dots , n \; .
\end{equation}

An important example for this paper is the Lotka-Volterra system $\LV(\Ga n)$, whose underlying graph $\Ga n$ has
vertices $1,2,\dots,n$ and has an arc from $i$ to $j$ with value $1$ when $i<j$ (see Figure~\ref{fig:Gamma_n}).

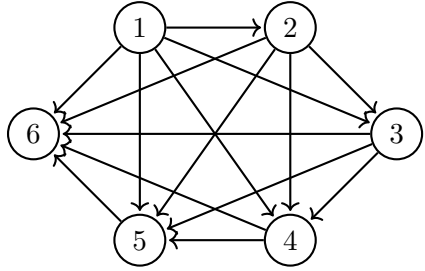
\begin{figure}[h]
  \begin{tikzpicture}[->,shorten >=1pt,auto,node distance=2cm,
                thick,main node/.style={circle,draw,font=\bfseries}]
  \node[main node] (1) {$1$};
  \node[main node] (2) [right of=1] {$2$};
  \node[main node] (3) [below right of=2] {$3$};
  \node[main node] (4) [below left of=3] {$4$};
  \node[main node] (5) [left of=4] {$5$};
  \node[main node] (6) [below left of=1] {$6$};
  \path
    (1) edge (2) edge(3) edge (4) edge (5) edge (6)
    (2) edge (3) edge (4) edge (5) edge (6)
    (3) edge (4) edge (5) edge (6)
    (4) edge (5) edge (6)
    (5) edge (6);
\end{tikzpicture}
\caption{The graph $\Ga n$ corresponds to the skew-symmetric $n\times n$ matrix whose upper-trangular entries
  are all equal to~$1$. The pictured graph is $\Ga 6$.\label{fig:Gamma_n}}
\end{figure}

\subsection{The Kahan-Poisson property}
The property of the Kahan map in which we are interested in this paper is its preservation of the Poisson structure
of the associated Lotka-Volterra system, i.e., that the Kahan map, which is a birational automorphism of $\bbF^n$,
is also a Poisson map (and hence a birational Poisson automorphism). Before giving the definition, let us clarify
the independence on $\varepsilon$: when the Kahan map is a Poisson map for some value of $\varepsilon$ then it is a
Poisson map for all values of $\varepsilon$. Indeed, the defining equations \eqref{eq:LV_Kahan} of the Kahan map
are homogeneous when $\varepsilon$ is given weight $-1$, while giving a weight $1$ to all $x_i$; the claim then
follows from the fact that homotheties of quadratic Poisson structures are Poisson maps (see \cite[Proposition
  8.16]{PLV}).  It also shows that if the Kahan map of a Lotka-Volterra system $\LV(\Gamma)$ is a Poisson map then
the Kahan map of the Lotka-Volterra system $\LV(\gamma\Gamma)$, where $\gamma\Gamma$ is $\Gamma$ with all of its
values scaled by $\gamma\in\bbF^*$, is also a Poisson map. We therefore set $\varepsilon =1$ and when we speak of
the Kahan map or the Kahan morphism of a Lotka-Volterra system, it is implictly assumed that
$\varepsilon=1$. Notice that when $\gamma=-1$, then $\gamma\Gamma=-\Gamma$ is the graph $\Gamma$ with the direction
of all its arcs reversed.

\begin{defn}%
  Let $\Gamma$ be a skew-symmetric graph. Then $\Gamma$ is said to have the \emph{Kahan-Poisson property} (or
  \emph{\KP}) if the Kahan map of its associated Lotka-Volterra system $\LV(\Gamma)$ is a Poisson map.
\end{defn}
In algebraic terms, this means that the Kahan morphism $K$ is an automorphism of Poisson fields, i.e.,
$\pb{K(x_i),K(x_j)}=K\pb{x_i,x_j}$, for all $i,j=1,\dots,n$, which can also be written as $\pb{\tx i,\tx
  j}=\widetilde{\pb{x_i,x_j}}$.

A first family of skew-symmetric graphs which have the \KP\ is given by the following proposition:
\begin{prop}\cite{KKQTV}\label{prp:LVn0_KP}
  For any $n\in\bbN^*$, the graph $\Ga n$ has the \KP.
\end{prop}
\begin{proof}
For the proof, we refer to \cite[Proposition 3.8]{KKQTV}. Yet, we point out the crucial fact that makes the
computation feasible. In terms of the variables $u_1,\dots,u_n$, and $\tilde u_1,\dots,\tilde u_n$, defined by
$u_i:=x_1+x_2+\cdots+x_i$, and $\tilde u_i:=\tx1+\tx2+\cdots+\tx i$, the Kahan map (with $\varepsilon=1$) takes the
simple separated form
\begin{equation}\label{eq:Kahan_u_i}
  \tilde u_i=u_i\frac{1+H}{1-H+2u_i}\;, \quad i=1,2, \dots , n \; ,
\end{equation}%
while the Poisson bracket takes the form
\begin{equation}\label{eq:pb_in_u}
\pb{u_i,u_j}=u_i(u_j-u_i)\;, \quad\hbox{ for } 1\leqslant i<j\leqslant n\;.
\end{equation}
It therefore suffices to verify that $\pb{\tilde u_i,\tilde u_j}=\tilde u_i(\tilde u_j-\tilde u_i)$, for
$1\leqslant i<j\leqslant n$, with $\tilde u_i$ given by~\eqref{eq:Kahan_u_i}, which is easily done by direct
computation, using~\eqref{eq:pb_in_u}.
\end{proof}
As we mentioned earlier, if we multiply the adjacency matrix $A$ of $\Gamma$ by any non-zero scalar, the property
for the corresponding graph of having the \KP\ is not affected. In particular, Proposition \ref{prp:LVn0_KP}
implies that $\gamma\Ga n$ has the \KP\ for any $n\in\bbN^*$ and for any $\gamma\in\bbF^*$.

\subsection{Cloning and decloning}\label{par:cloning}
The cloning of Lotka-Volterra systems, which was introduced in \cite{PPP_graphs}, is the inverse operation to
decloning, which we already recalled in the introduction. Let $\Gamma=(S,A)$ be a skew-symmetric graph with vertex
set $S=\set{1,2,\dots,n}$. Let $\wght$ be a \emph{weight vector} for $\Gamma$, i.e., $\wght$ is a function
$\wght:S\to\bbN^*$. The \emph{cloning} of $(\Gamma,\wght)$ is the skew-symmetric graph
$\Gamma^\wght=(S^\wght,A^\wght)$, constructed as follows: on the one hand, every vertex $i\in S$ gives rise
to~$\wght(i)$ vertices in $S^\wght$, which we denote by $(i,1),(i,2),\dots,(i,\wght(i))$. On the other hand, the
entries $a^\wght_{(i,k),(j,\ell)}$ of the (skew-symmetric) adjacency matrix $A^\wght$ of $\Gamma^\wght$ are defined
by $a^\wght_{(i,k),(j,\ell)}:=a_{i,j}$, for $i,j\in S$ and $1\leqs k\leqs\wght(i)$, $1\leqs \ell\leqs\wght(j)$. By
definition, the \emph{cloning} of $(\LV(\Gamma),\wght)$ is~$\LV(\Gamma^\wght)$. We will denote the linear
coordinates corresponding to the vertices $(i,k)$ by $y_{i,k}$ and write, as above, $\bbF(y)$ for the field of
rational functions $\bbF(y_{1,1},y_{1,2},\dots,y_{n,\wght(n)})$. The Poisson bracket on~$\bbF(y)$ associated with
$A^\wght$ is given by $\pb{y_{i,k},y_{j,\ell}}^\wght=a_{i,j}y_{i,k}y_{j,\ell}$, as follows from the definition
of~$\Gamma^\wght$. For $1\leqslant i\leqslant n$ and $1\leqslant k,m\leqslant \wght(i)$ the functions
$y_{i,m}/y_{i,k}$ are \emph{Casimir functions} of $\PB^\wght$: they belong to the center of the Poisson
bracket. As Hamiltonian, we again take the sum of all variables,
$H^\wght:=\sum_{i=1}^n\sum_{k=1}^{\wght(i)}y_{i,k}$. The Hamiltonian vector field $\X_{H^\wght}$ of
$\LV(\Gamma^\wght)$ is then given by the following differential equations:
\begin{equation*}
  \dot y_{i,k} = y_{i,k}\sum_{j=1}^n\sum_{\ell=1}^{\wght(j)}a_{i,j} y_{j,\ell}\;, \qquad \ i=1,\dots ,
  n,\ k=1,\dots,\wght(i)\; .
\end{equation*}
Its Kahan discretization is implicitly defined, as in \eqref{eq:LV_Kahan}, by the following equations:
\begin{equation}\label{eq:kahan_gen}
  \frac{\ty ik-y_{i,k}}\varepsilon = \ty ik\sum_{j=1}^n\sum_{\ell=1}^{\wght(j)}a_{i,j}
  y_{j,\ell}+y_{i,k}\sum_{j=1}^n\sum_{\ell=1}^{\wght(j)}a_{i,j}\ty j\ell\;,\qquad\ i=1,\dots,n,\ k=1,\dots,\wght(i)\;.
\end{equation}
The corresponding automorphism $K^\wght$ of $\bbF(y)$ is defined by
$K^\wght(y_{i,k}):=\ty ik$, for $i=1,\dots,n$ and $k=1,\dots,\wght(i)$. We view $\bbF(y)$ as a field extension of
$\bbF(x)$ using the \emph{decloning morphism} $\Sigma:\bbF(x)\to\bbF(y)$, defined by
\begin{equation}\label{equ:Sigma_def}
  \Sigma(x_i):=\sum_{k=1}^{\wght(i)}y_{i,k}\;.
\end{equation}
We show in the following proposition that the decloning morphism commutes with the Kahan morphism and that the KP
property is preserved under cloning and decloning:
\begin{prop}\label{prp:kahan_cloning}
  With the above definitions and notations, the following diagram of fields and field morphisms is commutative:
    \begin{equation}\label{dia:kahan_sigma_commute}
    \begin{tikzcd}[row sep=5ex, column sep=7ex]
      (\bbF(x),\PB)\arrow{r} {K}\arrow[swap]{d} {\Sigma}&(\bbF(x),\PB)\arrow{d}{\Sigma}\\
            {(\bbF(y),\PB^\wght)}\arrow[swap]{r}{K^\wght}&{(\bbF(y),\PB^\wght)}
    \end{tikzcd}
    \end{equation}
  The vertical arrow $\Sigma$ is a Poisson morphism, while $K$ is a Poisson morphism if and only if $K^\wght$ is a
  Poisson morphism.
\end{prop}
\begin{proof}
  We first show that the diagram is commutative. Let us set $\varepsilon=1$, as before. Using $\Sigma$,
  \eqref{eq:kahan_gen} can be rewritten as
\begin{equation}\label{eq:kahan_cloned}
  \ty ik-y_{i,k} = \ty ik\sum_{j=1}^na_{i,j}\Sigma(x_j)+y_{i,k}\sum_{j=1}^n a_{i,j}
  \overline{\Sigma(x_j)}\;.
\end{equation}
For fixed $i$, summing up \eqref{eq:kahan_cloned} for $k=1,\dots,\wght(i)$, we get
\begin{equation}\label{eq:kahan_cloned_summed}
  \overline{\Sigma(x_i)}-\Sigma(x_i) = \overline{\Sigma(x_i)}\sum_{j=1}^na_{i,j}\Sigma(x_j)+\Sigma(x_i)\sum_{j=1}^n
  a_{i,j} \overline{\Sigma(x_j)}\;.
\end{equation}
Recall that $\tx 1,\dots,\tx n$ is the unique solution to \eqref{eq:LV_Kahan}; if we write this solution
by making explicit its dependency on $x_1,\dots,x_n$ as $\tx i=R_i(x_1,\dots,x_n)$, then it follows from
comparing \eqref{eq:LV_Kahan} and~\eqref{eq:kahan_cloned_summed} that
$\overline{\Sigma(x_i)}=R_i(\Sigma(x_1),\dots,\Sigma(x_n))$. Since $\Sigma$ is an algebra homomorphism and $R_i$ is
a rational function of its arguments,
\begin{equation*}
  \overline{\Sigma(x_i)}=R_i(\Sigma(x_1),\dots,\Sigma(x_n))=\Sigma(R_i(x_1,\dots,x_n))=\Sigma(\tilde x_i)\;,
\end{equation*}%
showing the commutativity of the diagram.  For $1\leqslant i<j\leqslant n$ we have
$$
  \pb{\Sigma(x_i),\Sigma(x_j)}^\wght
  =\sum_{k=1}^{\wght(i)}\sum_{\ell=1}^{\wght(j)}\pb{y_{i,k},y_{j,\ell}}^\wght
  =a_{i,j}\sum_{k=1}^{\wght(i)}\sum_{\ell=1}^{\wght(j)}y_{i,k}y_{j,\ell}=a_{i,j}\Sigma(x_i)\Sigma(x_j)=\Sigma\pb{x_i,x_j}\;,
$$
as follows from the definitions of the Poisson brackets and of $\Sigma$. As a consequence, the vertical arrows in
the diagram \eqref{dia:kahan_sigma_commute} are morphisms of Poisson fields. In order to show that the two
horizontal arrows in that diagram are at the same time morphisms of Poisson fields, we first show that
\begin{equation}\label{eq:Casimirs}
  \Sigma(x_i)/y_{i,k} \hbox{ is a Casimir of }\PB^\wght \hbox{ and is an invariant of } K^\wght\;,
\end{equation}%
for $1\leqslant i\leqslant n$ and $1\leqslant k\leqslant \wght(i)$.  The first statement in \eqref{eq:Casimirs}
follows at once from the fact that $y_{i,m}/y_{i,k}$ is for any $1\leqslant k,m\leqslant \wght(i)$ a Casimir
function of~$\PB^\wght$, as we already recalled. The second statement means that
$\Sigma(x_i)/y_{i,k}=\overline{\Sigma(x_i)}/\bar y_{i,k}$. To prove the latter, divide \eqref{eq:kahan_cloned} by
$y_{i,k}$, to see that $\ty ik/y_{i,k}$ is independent of $k$; then $\ty ik/y_{i,k}=\ty im/y_{i,m}$, for $k$ and
$m$ as above, so that $y_{i,m}/y_{i,k}$ is an invariant, and hence also
$\Sigma(x_i)/y_{i,k}=\sum_{m=1}^{\wght(i)}y_{i,m}/y_{i,k}$. We can therefore write
\begin{equation}\label{eq:y_to_Y}
  \bar y_{i,k}=\frac{y_{i,k}}{\Sigma(x_i)}\overline{\Sigma(x_i)}\quad\hbox{ and }\quad
  \bar y_{j,\ell}=\frac{y_{j,\ell}}{\Sigma(x_j)}\overline{\Sigma(x_j)}\;.
\end{equation}%
In the second formula above, $1\leqslant j\leqslant n$ and $1\leqslant\ell\leqslant\wght(j)$. Then, 
\begin{align}\label{equ:Kahan_Poisson}
  \pb{\ty ik,\ty j\ell}^\wght&\stackrel{\eqref{eq:y_to_Y}}=
  \pb{\frac{y_{i,k}}{\Sigma(x_i)}\overline{\Sigma(x_i)},\frac{y_{j,\ell}}{\Sigma(x_j)}\overline{\Sigma(x_j)}}^\wght
  \stackrel{\eqref{eq:Casimirs}}=
  \frac{y_{i,k}y_{j,\ell}}{\Sigma(x_ix_j)}\pb{\overline{\Sigma(x_i)},\overline{\Sigma(x_j)}}^\wght\\
  &\stackrel{\eqref{eq:y_to_Y}}=\bar y_{i,k}\bar y_{j,\ell}
  \frac{\pb{\overline{\Sigma(x_i)},\overline{\Sigma(x_j)}}^\wght}{\overline{\Sigma(x_ix_j)}}
  \stackrel{\eqref{dia:kahan_sigma_commute}}= \bar y_{i,k}\bar
  y_{j,\ell}\frac{\pb{{\Sigma(\tx i)},{\Sigma(\tx j)}}^\wght}{\Sigma(\tx i\tx j)}
  =\bar y_{i,k}\bar y_{j,\ell}\Sigma\frac{\pb{\tx i,\tx j}}{\tx i\tx j}\;,\nonumber
\end{align}
where we have used in the last step that $\Sigma$ is a Poisson morphism. It follows that
\begin{align*}
  &\hbox{$K$ is a Poisson morphism}\\
  &\qquad\Longleftrightarrow\hbox{$\pb{\tx i,\tx j}=a_{i,j}\tx i\tx j$ for all $1\leqslant i,j\leqslant n$}\\
  &\qquad\stackrel{\eqref{equ:Kahan_Poisson}}\Longleftrightarrow
    \pb{\ty ik,\ty j\ell}^\wght=a_{i,j}\ty ik\ty j\ell\hbox { for all }
    1\leqslant i,j\leqslant n,\ 1\leqslant k\leqslant \wght(i),\ 1\leqslant \ell\leqslant \wght(j)\\
  &\qquad\Longleftrightarrow \hbox{$K^\wght$ is a Poisson morphism.}
\end{align*}
\end{proof}
Propositions \ref{prp:LVn0_KP} and \ref{prp:kahan_cloning} lead at once to the following corollary, which is the
inverse implication in Theorem \ref{thm:main}.

\begin{cor}\label{cor:inverse_implication}
  For any weight vector $\wght$ for $\Ga n$ and for any $\gamma\in\bbF^*$, the graph $\gamma\Ga n^\wght$ has the \KP.
  \qed
\end{cor}
\section{The 3-dimensional case}\label{sec:3dim}
In this section we prove the direct implication of Theorem \ref{thm:main} in case the skew-symmetric graph $\Gamma$
has $n=3$ vertices. We do not need to prove this in case $n=2$ because there is only one non-trivial skew-symmetric
graph with two vertices, which is the graph $\gamma\Ga 2$, with $\gamma\in\bbF^*$, for which we know from
Proposition \ref{prp:LVn0_KP} that it has the \KP.

\subsection{The known cases}\label{sec:known_dim3}
Let $\Gamma=(S,A)$ be any non-trivial skew-symmetric graph with three vertices, $S=\set{1,2,3}$. By assumption, it
has at least one arc, which we may suppose to be an arc between the vertices $1$ and~$2$, that is
$a_{1,2}\neq0$. Since, as we have seen, the \KP\ is preserved by a rescaling of $A$, we may suppose that
$a_{1,2}=1$; let us denote $\a:=a_{1,3}$ and~$\b:=a_{2,3}$. We list in Table \ref{tab:dim_3} the values of $\a$ and
$\b$ for which we know that the corresponding graph has the \KP, because it is isomorphic to $\Ga 3$ or to a
cloning of $\Ga2$ (Corollary \ref{cor:inverse_implication}). The last three graphs are isomorphic, but it will be
convenient to consider them all. For future use (Section~\ref{par:dim4}), notice also that these graphs $\Gamma$
are characterized amongst all non-trivial 3-vertex graphs as follows: $\Gamma$ has
\begin{enumerate}
  \item[$\bullet$] Either a single arc (in which case the graph is disconnected);
  \item[$\bullet$] Or two arcs, both starting from -- or ending in -- the same vertex;
  \item[$\bullet$] Or three arcs which do not form a circuit.
\end{enumerate}

\begin{table}[h]
  \def\arraystretch{3.5}
  \setlength\tabcolsep{0.4cm}
  \centering
\begin{tabular}{c|cccc}
   $(\a,\b)$&$\Gamma=\Delta^\wght$&$\Delta$&$\wght$&\\
  \hline 
  (0,0)
  &\raisebox{-3mm}{\tiny
    \begin{tikzpicture}[->,shorten >=1pt,auto,node distance=1cm,thick,
                                 main node/.style={circle,draw,font=\bfseries}]
      \node[main node] (1) {$1$};
      \node[main node] (3) [below right of=1] {$3$};
      \node[main node] (2) [above right of=3] {$2$};
      \path
      (1) edge (2) ;
    \end{tikzpicture}}
  &\raisebox{-3mm}{\tiny
    \begin{tikzpicture}[->,shorten >=1pt,auto,node distance=1cm,thick,
                                 main node/.style={circle,draw,font=\bfseries}]
      \node[main node] (1) {$1$};
      \node[main node] (3) [below right of=1] {$3$};
      \node[main node] (2) [above right of=3] {$2$};
      \path
      (1) edge (2) ;
    \end{tikzpicture}}
  &(1,1,1)&\\ 
  (0,-1)
  &\raisebox{-3mm}{\tiny
    \begin{tikzpicture}[->,shorten >=1pt,auto,node distance=1cm,thick,
                                 main node/.style={circle,draw,font=\bfseries}]
      \node[main node] (1) {$1$};
      \node[main node] (3) [below right of=1] {$3$};
      \node[main node] (2) [above right of=3] {$2$};
      \path
      (1) edge (2)
      (3) edge (2);
    \end{tikzpicture}}
  &\raisebox{-1mm}{\tiny
    \begin{tikzpicture}[->,shorten >=1pt,auto,node distance=1cm,thick,
                                 main node/.style={circle,draw,font=\bfseries}]
      \node[main node] (1) {$1$};
      \node[main node] (2) [right of=1] {$2$};
      \path
      (1) edge (2);
    \end{tikzpicture}}
  &(2,1)&\\
  (1,0)
  &\raisebox{-3mm}{\tiny
    \begin{tikzpicture}[->,shorten >=1pt,auto,node distance=1cm,thick,
                                 main node/.style={circle,draw,font=\bfseries}]
      \node[main node] (1) {$1$};
      \node[main node] (3) [below right of=1] {$3$};
      \node[main node] (2) [above right of=3] {$2$};
      \path
      (1) edge (2) edge (3);
    \end{tikzpicture}}
  &\raisebox{-1mm}{\tiny
    \begin{tikzpicture}[->,shorten >=1pt,auto,node distance=1cm,thick,
                                 main node/.style={circle,draw,font=\bfseries}]
      \node[main node] (1) {$1$};
      \node[main node] (2) [right of=1] {$2$};
      \path
      (1) edge (2);
    \end{tikzpicture}}
  &(1,2)&\\
  (1,-1)
  &\raisebox{-3mm}{\tiny
    \begin{tikzpicture}[->,shorten >=1pt,auto,node distance=1cm,thick,
                                 main node/.style={circle,draw,font=\bfseries}]
      \node[main node] (1) {$1$};
      \node[main node] (3) [below right of=1] {$3$};
      \node[main node] (2) [above right of=3] {$2$};
      \path
      (1) edge (2) edge (3)
      (3) edge (2);
    \end{tikzpicture}}
  &\raisebox{-3mm}{\tiny
    \begin{tikzpicture}[->,shorten >=1pt,auto,node distance=1cm,thick,
                                 main node/.style={circle,draw,font=\bfseries}]
      \node[main node] (1) {$1$};
      \node[main node] (3) [below right of=1] {$3$};
      \node[main node] (2) [above right of=3] {$2$};
      \path
      (1) edge (2) edge (3)
      (3) edge (2);
    \end{tikzpicture}}
  &(1,1,1)&\\
  (1,1)
  &\raisebox{-3mm}{\tiny
    \begin{tikzpicture}[->,shorten >=1pt,auto,node distance=1cm,thick,
                                 main node/.style={circle,draw,font=\bfseries}]
      \node[main node] (1) {$1$};
      \node[main node] (3) [below right of=1] {$3$};
      \node[main node] (2) [above right of=3] {$2$};
      \path
      (1) edge (2) edge (3)
      (2) edge (3);
    \end{tikzpicture}}
  &\raisebox{-3mm}{\tiny
    \begin{tikzpicture}[->,shorten >=1pt,auto,node distance=1cm,thick,
                                 main node/.style={circle,draw,font=\bfseries}]
      \node[main node] (1) {$1$};
      \node[main node] (3) [below right of=1] {$3$};
      \node[main node] (2) [above right of=3] {$2$};
      \path
      (1) edge (2) edge (3)
      (2) edge (3);
    \end{tikzpicture}}
  &(1,1,1)&\\
 (-1,-1)
  &\raisebox{-3mm}{\tiny
    \begin{tikzpicture}[->,shorten >=1pt,auto,node distance=1cm,thick,
                                 main node/.style={circle,draw,font=\bfseries}]
      \node[main node] (1) {$1$};
      \node[main node] (3) [below right of=1] {$3$};
      \node[main node] (2) [above right of=3] {$2$};
      \path
      (1) edge (2) 
      (3) edge (2) edge (1);
    \end{tikzpicture}}
  &\raisebox{-3mm}{\tiny
    \begin{tikzpicture}[->,shorten >=1pt,auto,node distance=1cm,thick,
                                 main node/.style={circle,draw,font=\bfseries}]
      \node[main node] (1) {$1$};
      \node[main node] (3) [below right of=1] {$3$};
      \node[main node] (2) [above right of=3] {$2$};
      \path
      (1) edge (2)
      (3) edge (2) edge (1);      
    \end{tikzpicture}}
  &(1,1,1)&\\
\end{tabular}\quad
\medskip
\caption{The six cases of non-trivial three-vertex graphs~$\Gamma$, with $a_{1,2}=1$, which we already know to have
  the Kahan-Poisson property. Every arc has value $1$.}\label{tab:dim_3}
\end{table}

We will prove that there are no other three-vertex graphs $\Gamma=(S,A)$, with $a_{1,2}=1$, which have the \KP.

\subsection{Computing efficiently the Poisson brackets}\label{par:compute_poisson}
Since the Kahan map of a Lotka-Volterra system is given by rational functions, which are already quite complicated
in dimension 3, we explain here how the condition that the Kahan morphism is a Poisson morphism leads to necessary
conditions that are computable by hand, and which will actually be sufficient for our purposes. Since we will use
our method also in dimension 4, we will explain it for any skew-symmetric graph $\Gamma$ with $n>2$ vertices. As
before, the adjacency matrix of $\Gamma$ is denoted by $A$. We first write down the basic equations and introduce
some notation. We write the equations~\eqref{eq:LV_Kahan} in two different ways (recall that we have set
$\varepsilon=1$). On the one hand, we write them as a linear system in $\tx 1,\dots,\tx n$,
\begin{equation}\label{eq:Kahan_as_linear}
  M\tilde x^T=x^T\;.
\end{equation}%
Notice that every entry of $M$ is an affine function of $x_1,\dots,x_n$. For a point $P\in\bbF^n$ we denote by
$M(P)$ the evaluation of $M$ at $P$ and by $\tP$ the image of~$P$ under the Kahan map, i.e.,
$\tP=(\tx1(P),\tx2(P),\tx3(P))$. On the other hand, it is easy to see that \eqref{eq:LV_Kahan} can also be written
as $\cL_k=0$, for $k=1,2,\dots,n$, where $\cL_k$ is defined by
\begin{equation}\label{equ:L_i}
  \cL_k:=\(\sum_{j=1}^na_{k,j}x_j-1\)\tx k+x_k\sum_{j=1}^na_{k,j}\tx j+x_k\;.
\end{equation}%
Indeed, \eqref{equ:L_i} is the $k$-th entry of $x^T-M\tilde x^T$. We denote the column vector whose $k$-th entry is
$\cL_k$ by $\cL$ and write $\cL(\tP)$ for $\cL$ with the functions $\tx 1,\dots,\tx n$ evaluated at $P$. Each
entry of $\cL(\tP)$ is also an affine function of $x_1,\dots,x_n$.

We can now explain how to explicitly compute, for given $i,j$ and $P$, the condition $\pb{\tx i,\tx
  j}(P)=a_{i,j}\tx i(P)\tx j(P)$. We do this in four different steps.

\underline{Step 1:} The image point $\tP$.

If $M(P)$ is non-singular, i.e., $\det M(P)\neq0$, the Kahan map is defined at~$P$ and its image $\tP$ is computed
from
\begin{equation}\label{equ:Pt}
  M(P)\tP^T=P^T\;.
\end{equation}%
Notice that $\det M(P)$ depends (linearly) on the entries $a_{k,\ell}$ of $A$, so for a given $P$ it may be zero
for some values of these entries $a_{k,\ell}$. The computations that follow are then not valid for these values; as
we will see, it is important to keep track of these values. In the steps which follow, we assume that $\det
M(P)\neq0$.

\underline{Step 2:} Computation of the Poisson brackets $\pb{\tx j,x_\ell}(P)$, $\ell=1,\dots,n$.

We compute the Poisson brackets $\pb{\tx k,x_\ell}(P)$ for $k,\ell=1,\dots,n$. This can easily be done directly
from \eqref{eq:Kahan_as_linear} without solving the latter for $\tilde x$, as follows. Let us denote by
$\pb{M,x_\ell}$ the matrix obtained by taking the Poisson bracket of every entry of $M$ with $x_\ell$, and
similarly for the column vector $\pb{\tilde x^T,x_\ell}$. Then it follows from \eqref{eq:Kahan_as_linear}, using
the Leibniz rule, that $ \pb{M,x_\ell}\tilde x^T+M\pb{\tilde x^T,x_\ell}=\pb{x^T,x_\ell}, $
so that, at $P$,
\begin{equation}\label{equ:first_pb}
  M(P)\pb{\tilde x^T,x_\ell}(P)=\pb{x^T,x_\ell}(P)-\pb{M,x_\ell}(P)\tP^T\;.
\end{equation}%
This gives a linear system for the brackets $\pb{\tx k,x_\ell}(P)$, $k=1,\dots,n$. Notice that the matrix governing
the linear system is again $M(P)$, so that these brackets are uniquely determined from \eqref{equ:first_pb} (recall
that we have assumed that $\det M(P)\neq0$). Also, the right hand side of \eqref{equ:first_pb} is equal to
$\pb{\cL(\tP),x_\ell}(P)$, since $\cL_k$ is the $k$-th entry of $x^T-M\tilde x^T$ (see
\eqref{eq:Kahan_as_linear}). It means that, in order the compute the right hand side of \eqref{equ:first_pb}, we
can start from the equations \eqref{eq:LV_Kahan} defining the Kahan map, evaluate the functions $\tx 1,\dots,\tx n$
at~$P$, and then take the Poisson bracket at $P$ of the remaining affine functions in $x_1,\dots,x_n$ with
$x_\ell$. Doing this for $\ell=1,\dots,n$ and solving the resulting linear system, we find the brackets $\pb{\tx
k,x_\ell}(P)$ for $k,\ell=1,\dots,n$. They are rational functions of the entries of the adjacency matrix $A$ of
$\Gamma$.

\underline{Step 3:} Computation of the Poisson brackets $\pb{\tx k,\tx j}(P)$.

The Poisson brackets $\pb{\tx k,\tx j}(P)$ for $k=1,\dots,n$ are computed in a quite similar way, using the $n$
Poisson brackets which were computed in Step 2 (recall that $j$ is fixed). In this step we take the Poisson bracket
of \eqref{eq:Kahan_as_linear} with $\tx j$ at $P$ to obtain, as in Step 2,
\begin{equation}\label{equ:second_pb}
  M(P)\pb{\tilde x^T,\tx j}(P)=\pb{x^T,\tx j}(P)-\pb{M,\tx j}(P)\tP^T=\pb{\cL(\tP),\tx j}(P)\;.
\end{equation}%
Notice that, again, the defining matrix of the linear system is $M(P)$ and that, again, the right hand side can
easily be computed from the equations defining the Kahan map, where the functions $\tx 1,\dots,\tx n$ are evaluated
at~$P$; it is here that one needs the Poisson brackets $\pb{\tx j,x_\ell}$, for $\ell=1,2,\dots,n$, which were
computed in Step~2.  Solving the resulting linear system, we find the brackets $\pb{\tx k,\tx j}(P)$ for
$k=1,\dots,n$. Again, they are rational functions of the entries of the adjacency matrix $A$ of $\Gamma$.

\underline{Step 4:} The Poisson morphism condition(s).

From the previous step we know $\pb{\tx i,\tx j}(P)$ and we can now write down explicitly the condition that
\begin{equation*}
  \pb{\tx i,\tx j}(P)=a_{i,j}\tx i(P)\tx j(P)\;,
\end{equation*}%
which is a sufficient condition for the Kahan morphism to be a Poisson morphism. Since the left hand side is a
rational function of the entries of the adjacency matrix $A$, we get a rational condition on these entries. If the
condition is not satisfied, we can conclude that the graph $\Gamma$ does not have the \KP. This is how we will use
this condition in what follows.

\subsection{The $3$-dimensional case}
We now prove the direct implication in Theorem \ref{thm:main} in the three-dimensional case. We assume that
$\Gamma=(S,A)$ is a three-vertex graph, with $a_{1,2}=1$. We need to show that when $\Gamma$ has the \KP,
which is equivalent to saying that
\begin{equation}\label{equ:Poisson_condition}
  \pb{\tx i,\tx j}=a_{i,j}\tx i\tx j\;,\quad \hbox{ for } 1\leqslant i<j\leqslant 3\;,
\end{equation}%
then $\Gamma$ is one of the graphs $\Gamma^\wght$ in Table \ref{tab:dim_3}. We will do this by
computing~\eqref{equ:Poisson_condition} at some well-chosen points (and for some particular values of $i,j$), using
the method of the previous subsection.

To do this, it is helpful to represent the six values from the first column in Table \ref{tab:dim_3} as points in
the plane, as indicated in Figure \ref{fig:dim_3}. One sees from the figure that these six points lie on (one or
two of) the lines $\a=1,\ \b=-1, \ \a=\b$. This will guide us in the proof as follows: we will first show that for
points $(\a,\b)$ on one of these three lines, the corresponding graph can only have the \KP\ if it is one of the
three special points on that line. Secondly, we will show that for points $(\a,\b)$ not lying on any of these
lines, the corresponding graph cannot have the \KP.

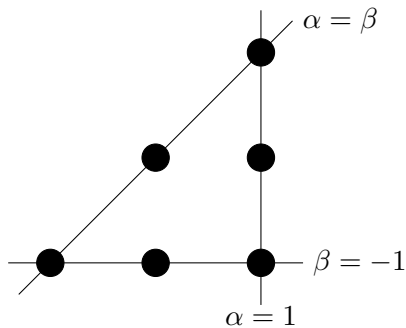
\begin{figure}[h]
\begin{tikzpicture}[scale=1.4,main node/.style={circle,draw,fill, font=\bfseries}]
  \draw (-1.4,-1) -- (1.4,-1) node (x)[right]{$\beta=-1$};
  \draw (1,-1.4) -- (1,1.4) ;
  \draw (-1.3,-1.3) -- (1.3,1.3) node (z)[right]{$\alpha=\beta$};
  \node (y) at (1,-1.5) {$\alpha=1$}; 
  \node[main node] (1) at (1,1){};
  \node[main node] (2) at (0,0){}; 
  \node[main node] (3) at (-1,-1){};
  \node[main node] (4) at (0,-1){};
  \node[main node] (5) at (1,0){};
  \node[main node] (6) at (1,-1){};
\end{tikzpicture}
\caption{Viewed as points in the plane, the six values $(\a,\b)$ for which we know that the corresponding 3-vertex
  graph has the KP property lie on three lines, each of which contains three of the points.\label{fig:dim_3}}
\end{figure}

We start with the line $\a=1$, so we suppose that $(\a,\b)=(1,\b)$, and we take $P:=(1,-1,-1)$. The condition which
we will compute is $\pb{\tx1,\tx3}(P)=\tx1(P)\tx3(P)$, so that $i=1$ and $j=3$. Since~$a_{1,2}=a_{1,3}=1$ and
$a_{2,3}=\b$, the Kahan map is in this case defined by the following linear system:
\begin{align}\label{equ:system_a=1}
  \tx 1 -x_1 &= \tx1(x_2+x_3)+x_1(\tx 2+\tx 3)\;,\nonumber\\
  \tx 2 -x_2 &= \tx2(\b x_3-x_1)+x_2(\b\tx 3-\tx 1)\;,\\
  \tx 3 -x_3 &= \tx3(-x_1-\b x_2)-x_3(\tx 1+\b\tx 2)\;.\nonumber
\end{align}
Following Step 1 of the method, we write these equations as $M\tilde x^T=x^T$, where
\begin{equation*}
  M=  \begin{pmatrix}
    1-x_2-x_3&-x_1&-x_1\\
    x_2&1+x_1-\b x_3&-\b x_2\\
    x_3&\b x_3&1+x_1+\b x_2
  \end{pmatrix}\;.
\end{equation*}
Then
\begin{equation*}
  M(P)=  \begin{pmatrix}
    3&-1&-1\\
    -1&2+\b&\b\\
    -1&-\b &2-\b
  \end{pmatrix},\;
  M(P)^{-1}=\frac18
  \begin{pmatrix}
    4&2&2\\
    2(1-\b)&5-3\b&1-3\b\\
    2(\b+1)&1+3\b&5+3\b
  \end{pmatrix}\;,
\end{equation*}
so that $\tP=-\frac12(0,1-\b,1+\b)$. Note that $\det M(P)=8$, so that $\tP$ is defined for any value of~$\b$.
Evaluating $\tx 1,\tx 2,\tx 3$ in \eqref{equ:system_a=1} at $P$, leads to
\begin{equation*}
  \cL(\tP)=\frac12\begin{pmatrix}
    0\\
    (1-\b)(x_1+(2+\b)x_2-\b x_3+1)\\
    (1+\b)(x_1+\b x_2+(2-\b)x_3+1)\\
  \end{pmatrix}\;.
\end{equation*}%
It follows that $\pb{\cL(\tP),x_\ell}(P)$ is, for $\ell=1,2,3$, respectively given by
\begin{equation*}
  \begin{pmatrix}
  0\\1-\b\\1+\b
  \end{pmatrix}\;,\quad
  \frac{1+\b}2\begin{pmatrix}
  0\\-(1-\b)^2\\(1-\b)^2-2
  \end{pmatrix}\quad\hbox{and}\quad
  \frac{\b-1}2\begin{pmatrix}
  0\\2-(1+\b)^2\\(1+\b)^2
  \end{pmatrix}\;.
\end{equation*}%
According to Step 2, we get the Poisson brackets $\pb{\tx k,x_\ell}(P)$ by multiplying these three vectors with the
inverse of $M(P)$. We display only the brackets $\pb{x_\ell,\tx 1}(P)$ since they are the only ones needed to
compute $\pb{\cL(\tP),\tx1}(P)$ and to finish the computation:
\begin{equation*}
  \pb{x_1,\tx1}(P)=-\frac12\;,\quad \pb{x_2,\tx1}(P)=\frac{1+\b}4\;,\quad  \pb{x_3,\tx1}(P)=\frac{1-\b}4\;.
\end{equation*}%
Using these values we get
\begin{equation*}
  \pb{\cL(\tP),\tx1}(P)=\frac{\b(\b^2-1)}4
  \begin{pmatrix}
    0\\-1\\1
  \end{pmatrix}\;.
\end{equation*}%
According to Step 3, we get the Poisson brackets $\pb{\tx k,\tx 1}(P)$ by multiplying this vector with the inverse
of $M(P)$. The only bracket we need is $\pb{\tx1,\tx3}(P)$, which is found to be equal to
$\frac18\b(1-\b^2)$. Since $\tx1(P)=0$ the condition that $\pb{\tx1,\tx3}(P)=\a\tx1(P)\tx3(P)$ reduces to the
condition on~$\b$ that $\b(1-\b)(1+\b)=0$. We have therefore shown that if $(\a,\b)$ is on the line $\a=1$, the
corresponding graph can only have the \KP\ if $\b\in\set{-1,0,1}$, which corresponds precisely to the three values
for which we know that the corresponding graph has the \KP.

We next consider the line $\b=-1$ and we take $P:=(1,-1,1)$. We have that
\begin{equation*}
  M(P)=  \begin{pmatrix}
    2-\a&-1&-\a\\
    -1&3&-1\\
    \a&-1 &\a+2
  \end{pmatrix}\;,
  \quad\hbox{ so that }
  \quad
  \tP=\frac12\begin{pmatrix}
  1+\a\\0\\1-\a
  \end{pmatrix}\;.
\end{equation*}
Again, $\det M(P)=8$ and $\tP$ is defined for any value of $\a$. The computations of the brackets are very similar
to the previous case and leads to a similar result: since $\pb{\tx1,\tx2}(P)=\frac a8(1-a^2)$ and $\tx2(P)=0$, we
may conclude as in the first case from the equality $\pb{\tx1,\tx2}(P)=\tx1(P)\tx2(P)$ that if $(\a,\b)$ is on the
line $\b=-1$, the corresponding graph can only have the \KP\ if $\a\in\set{-1,0,1}$, which corresponds again
precisely to the three values for which we know that the corresponding graph has the \KP.

We move to the case that $(\a,\b)$ is on the line $\a=\b$. We choose $P=(-\frac12,1,\frac12)$. A new phenomenon is
now that $\det M(P)$ depends on $\a$, since
\begin{equation*}
  M(P)= \frac12\begin{pmatrix}
    -\a&1&\a\\
    2&1-\a&-2\a\\
    \a&\a &\a+2
  \end{pmatrix}\;,
  \quad\det M(P)=-\frac12(\a+1)\;.
\end{equation*}
It means that what follows is not valid for $\a=-1$, but that is not a problem since $(-1,-1)$ is precisely one of
the cases for which we know that the corresponding graph has the \KP. Thus, under this assumption, we may continue
as before and compute the Poisson brackets to find that $\pb{\tx 1,\tx 2}(P) =\frac{\a-1}4\((\a-1)^2+5\)$. The
equality $\pb{\tx1,\tx2}(P)=\tx1(P)\tx2(P)$ then evaluates to $\a^2(\a-1)=0$, showing that if $\a=\b$, with
$\a\neq-1$, then the graph corresponding to $(\a,\b)$ can only have the \KP\ if $\a=0$ or $\a=1$, as was to be
shown.

We finally consider the case where $(\a,\b)$ does not belong to any of the three red lines of
Figure~\ref{fig:dim_3}, i.e., we suppose that $\a\neq\b$, that $\a\neq1$ and that $\b\neq-1$. We show that the
corresponding graph cannot have the \KP. As in the previous cases we use the method of Section
\ref{par:compute_poisson}. The formulas are slightly more complicated, so we present the computation in some more
detail.

We choose the point $P:=(-1,1,1)$. Then
\begin{equation*}
  \pb{x_1,x_2}(P)=-1\;,\quad  \pb{x_1,x_3}(P)=-\a\;,\quad  \pb{x_2,x_3}(P)=\b\;.
\end{equation*}%
Also,
\begin{equation*}
  M(P)=
  \begin{pmatrix}
    -\a&1&\a\\
    1&-\b&-\b\\
    \a&\b&1-\a+\b
  \end{pmatrix}\;,
  \hbox{ with }
  \delta:=\det M(P)=(\a-1)(\b+1)\;,
\end{equation*}%
and
\begin{equation*}
  M(P)^{-1}=\frac1\delta
  \begin{pmatrix}
    (\a-1)\b&(\a-1)(\b+1)&(\a-1)\b\\
    \a-\a\b-\b-1&-\a(\b+1)&\a(1-\b)\\
    (\a+1)\b&\a(\b+1)&\a\b-1
  \end{pmatrix}\;.
\end{equation*}%
It follows that $\tP^T=M(P)^{-1}(-1,1,1)^T=(1,-1,1)^T,$ so that
\begin{equation*}
  \cL(\tP)=
  \begin{pmatrix}
    -1+\a x_1+x_2+\a x_3\\
    1+ x_1+\b x_2-\b x_3\\
    -1-\a x_1-\b x_2+(1-\a+\b) x_3
  \end{pmatrix}\;,
\end{equation*}
and
$$
  \pb{\cL(\tP),x_j}(P)=
  \begin{pmatrix}
    1+\a^2&-\a(1+\b)&\b-\a^2\\
    (1-\a)\b&\b^2-1&\b^2-\a\\
    (1-\a)(\a-\b)&(1+\b)(\a-\b)&\a^2-\b^2
  \end{pmatrix}\;.
$$
Multiplying the latter matrix with $M(P)^{-1}$ we obtain the matrix of Poisson brackets $\pb{\tx i,x_j}(P)$ of
which we only need the first line, which yields the following Poisson brackets:
$$
  \pb{\tx 1,x_1}(P)=\frac{2\b}{\b+1}\;,\
  \pb{\tx 1,x_2}(P)=-1\;,\
  \pb{\tx 1,x_3}(P)=\frac{2\b^2-\a\b-\a}{\b+1}\;.
$$
Indeed, they are sufficient to compute the Poisson bracket $\pb{\cL(\tP),\tx 1}(P)$, which is given by
%
\begin{equation}\label{equ:Ltx}
  \pb{\cL(\tP),\tx 1}(P)=\frac1{\b+1}
  \begin{pmatrix}
    \a&1&\a\\
    1&\b&-\b\\
    -\a&-\b&1-\a+\b
  \end{pmatrix}
  \begin{pmatrix}
    -2\b\\ \b+1\\ \a+\a\b-2\b^2
  \end{pmatrix}
  =\begin{pmatrix}
  1-2\a\b+\a^2\\
  \b(2\b-1-\a)\\
  (\a-\b)(2\b-\a+1)
  \end{pmatrix}\;.
\end{equation}
We can now compute $\pb{\tx 3,\tx 1}(P)$ as the third entry of
$\pb{\tilde{x}^T,\tx1}(P)=M(P)^{-1}\pb{\cL(\tP),\tx1}(P)$, i.e., as the product of the third line of $M(P)^{-1}$
and \eqref{equ:Ltx}. After some simplifications, we get
\begin{equation*}
  \pb{\tx3,\tx1}(P)=\frac1\delta(\a-1)(\b+1)(\a-2\b)=\a-2\b\;,
\end{equation*}%
so that
\begin{equation*}
  \pb{\tx3,\tx1}(P)+\a\tx1(P)\tx3(P)=2(\a-\b)\neq0\;,
\end{equation*}%
since it was supposed that $\a\neq\b$. This shows our claim.

\section{The higher-dimensional case}\label{sec:ndim}
In this section we prove Theorem \ref{thm:main} when $\Gamma$ has at least four vertices. To do this, we will use
reduction, which amounts to removing vertices from the graph. In Section~\ref{par:reduction} we recall reduction
and show that the \KP\ is preserved by reduction. This is used in Section~\ref{par:dim4} to prove Theorem
\ref{thm:main} when $\Gamma$ has four vertices, and in Section~\ref{par:dimn} to prove Theorem~\ref{thm:main}
when~$\Gamma$ has more than four vertices.

\subsection{Reduction}\label{par:reduction}
Lotka-Volterra systems admit a class of natural reductions, which we first recall. Let $\Gamma=(S,A)$ be a graph
with vertex set $S=\set{1,2,\dots,n}$, where $n>1$. Let $S'\subset S$ be a proper, non-empty subset and denote
$m:=\#S'$. We denote by $\Gamma'=(S',A')$ the induced subgraph of~$\Gamma$. The natural inclusion map
$\Gamma'\to\Gamma$ is a graph morphism, hence induces a morphism of Lotka-Volterra systems
$\LV(\Gamma')\to\LV(\Gamma)$ (see \cite[Prop.~3.2]{PPP_graphs}). If we denote the natural coordinates on these
Lotka-Volterra systems respectively by $y=(y_1,\dots,y_m)$ and $x=(x_1,\dots,x_n)$, then the induced morphism
$\imath:\bbF(y)\to\bbF(x)$ is an injective Poisson morphism; if we denote by~$\tau$ the unique strictly increasing
function $\tau:\set{1,\dots,m}\to\set{1,\dots,n}$ which takes values in $S'$, then $\imath(y_i)=x_{\tau(i)}$, for
$i=1,\dots,m$. Let us denote by $K:\bbF(x)\to\bbF(x)$ and by $K':\bbF(y)\to\bbF(y)$ the Kahan morphisms of
$\LV(\Gamma)$, respectively of $\LV(\Gamma')$. Then $K'$ is the restriction of $K$ to $\bbF(y)$, i.e., the
following diagram is commutative:
\begin{equation}\label{dia:kahan_reduction_commute}
  \begin{tikzcd}[row sep=5ex, column sep=7ex]
    \bbF(y)\arrow{r}{K'}\arrow[swap]{d} {\imath}&\bbF(y)\arrow{d}{\imath}\\
          {\bbF(x)}\arrow[swap]{r}{K}&{\bbF(x)}
  \end{tikzcd}
\end{equation}
To see this, it suffices to consider the equations \eqref{eq:LV_Kahan} which define $K$ and notice that upon
setting $x_j=\tx j=0$ for all $j\in S\setminus S'$, only the equations corresponding to $i\in S'$ remain and are
the equations which define the Kahan morphism $K'$.

Suppose now that $\Gamma$ has the \KP, so that $K$ is a Poisson morphism. Since $\imath$ is an injective Poisson
morphism, we may conclude as in the first part of the proof of Proposition \ref{prp:kahan_cloning} that $K'$ is
also a Poisson morphism. We state this as the following result.
\begin{prop}\label{prp:reduction}
  Suppose that $\Gamma$ is a skew-symmetric graph which has the \KP. Then any induced subgraph $\Gamma'$ of
  $\Gamma$ has the \KP.\qed
\end{prop}
One consequence of this proposition is immediate: if $\Gamma$ is a connected graph which satisfies the \KP, then up
to a sign all arcs have the same value.

We will also use the following lemma, which characterizes the graphs $\Gamma_n$. Recall that a \emph{tournament} is
an (unvalued) graph having (precisely) one arc between any pair of different vertices.

\begin{lemma}\label{lma:tournaments}
  Suppose that $\Gamma$ is a tournament with $n$ vertices, having the property that none of its 3-vertex induced
  subgraphs (triangles) is a circuit. Then $\Gamma\simeq\Gamma_n$.
\end{lemma}
\begin{proof}
Recall that, by definition, $\Gamma_n$ has vertices $1,2,\dots,n$ and has an arc from $i$ to $j$ (with value $1$)
when $i<j$.  For two vertices $v,w$ of $\Gamma=(S,A)$, let us write $v<w$ if there is an arc from $v$ to $w$. Then
$<$ is a total order relation on $S$: skew-symmetry is clear, as well as the fact that any pair of vertices is
comparable (since $\Gamma$ is a tournament), so we only need to prove transitivity. Let $u,v,w\in S$ and suppose
that $u<v$ and $v<w$. Then there is an arc from $u$ to $v$ and an arc from $v$ to $w$, so there cannot be an arc
from $w$ to $u$, since otherwise the subgraph induced by $\set{u,v,w}$ would be a circuit. Therefore, there is an
arc from $u$ to $w$ and $u<w$. Since all total orders on a set of $n$ elements are isomorphic, $\Gamma$ is
isomorphic to $\Gamma_n$ (by a unique isomorphism).
\end{proof}
As an immediate corollary of the lemma, we find that if $\Gamma$ is a tournament which has the KP
property, then $\Gamma\simeq\Gamma_n$, where $n$ is the number of vertices of $\Gamma$.

\subsection{The 4-dimensional case}\label{par:dim4}
We use Proposition \ref{prp:reduction} to prove Theorem~\ref{thm:main} when the (connected) graph $\Gamma$ has four
vertices. In view of Proposition \ref{prp:kahan_cloning}, we may suppose that $\Gamma$ is irreducible. Finally, we
may assume in view of the comment following Proposition \ref{prp:reduction} that all arcs of $\Gamma$ have
value~1, i.e., that $\Gamma$ is not valued. We show that the only such graph~$\Gamma$ which has the \KP\ property
is $\Gamma_4$.

We know from Proposition \ref{prp:reduction} that if we remove a vertex from $\Gamma$ then the remaining graph
$\Gamma'$ should have the \KP, so according to Section~\ref{sec:known_dim3} it is trivial or it has either
\begin{enumerate}
  \item[(I)] A single arc (in which case the graph is disconnected);
  \item[(II)] Two arcs, both starting from -- or ending in -- the same vertex;
  \item[(III)] Three arcs which do not form a circuit.
\end{enumerate}
We first show that any connected irreducible 4-vertex graph having one of these three properties is -- up to a reversal
of the direction of all arcs -- isomorphic to one of the three graphs in Figure \ref{fig:dim4}.
\begin{figure}[h]
{\footnotesize
    \begin{tikzpicture}[->,shorten >=1pt,auto,node distance=1cm,thick,
                                 main node/.style={circle,draw,font=\bfseries}]
      \node[main node] (a) {};
      \node[main node] (v) [below right of=a] {};
      \node[main node] (b) [above right of=v] {};
      \node[main node] (c) [below of=v] {};
      \path
      (a) edge (b) edge (c) edge (v)
      (b) edge (c) edge (v)
      (c) edge (v);
\end{tikzpicture}}
\qquad\qquad
{\footnotesize
    \begin{tikzpicture}[->,shorten >=1pt,auto,node distance=1cm,thick,
                                 main node/.style={circle,draw,font=\bfseries}]
      \node[main node] (a) {};
      \node[main node] (v) [below right of=a] {};
      \node[main node] (b) [above right of=v] {};
      \node[main node] (c) [below of=v] {};
      \path
      (a) edge (b) edge (v)
      (b) edge (v)
      (c) edge (v);
\end{tikzpicture}}
\qquad\qquad
{\footnotesize
    \begin{tikzpicture}[->,shorten >=1pt,auto,node distance=1cm,thick,
                                 main node/.style={circle,draw,font=\bfseries}]
      \node[main node] (a) {};
      \node[main node] (v) [below right of=a] {};
      \node[main node] (b) [above right of=v] {};
      \node[main node] (c) [below of=v] {};
      \path
      (a) edge (b) edge (v) 
      (c) edge (v);
\end{tikzpicture}}
\caption{Up to isomorphism and modulo a reversal of all arcs, there are only three connected irreducible
  4-vertex graphs for which every 3-vertex induced subgraph has the \KP. The first one pictured is
  $\Gamma_4$.\label{fig:dim4}}
\end{figure}
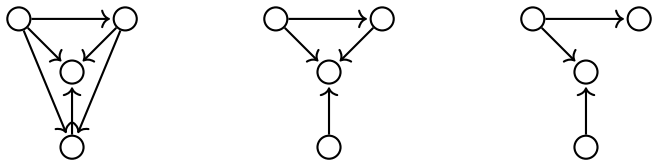
Since $\Gamma$ is connected, $\Gamma$ contains at least one vertex of degree two or three. More precisely, there
are the following three (disjoint) possibilities:
\begin{enumerate}
  \item[(i)] All vertices of $\Gamma$ have degree three;
  \item[(ii)] $\Gamma$ has a vertex of degree three and three vertices of degree one;
  \item[(iii)] $\Gamma$ has a vertex of degree $2$.
\end{enumerate}
We first consider the case (i). Then $\Gamma$ is a tournament which does not contain a circuit. According to
Lemma~\ref{lma:tournaments}, $\Gamma\simeq\Gamma_4$. It corresponds to the first graph in Figure \ref{fig:dim4}.

We next consider the case (ii), so $\Gamma$ has one vertex of degree three and three vertices of degree
one. If one removes any vertex of degree one, the resulting 3-vertex graph must be of type
(II), hence the three arcs must be pointing toward the vertex of degree three, or away from
it. In any case, $\Gamma$ is reducible.

We now consider the case (iii), in which we will need to consider several subcases. By assumption, $\Gamma$ has a
vertex $v$ of degree two. We call $\Gamma'$ the graph obtained by removing $v$ from~$\Gamma$ (together with the
arcs incident with~$v$). Since $\Gamma$ is connected, $\Gamma'$ is non-trivial, so it is either of type (I), (II)
or (III).  We analyse each of them separately. We start with type (III) and distinguish three cases, according to
whether or not $v$ is connected to the unique vertex with in and outdegree 1 (called $b$ in Figure
\ref{fig:dim4_1}). In each one of these cases, there is a unique way to add the other arc(s) incident to $v$; the
latter arc is indicated in blue. It is clear from the figure that each one of these cases is reducible: the two
vertices which are not connected have the same (in and out) neighbors.
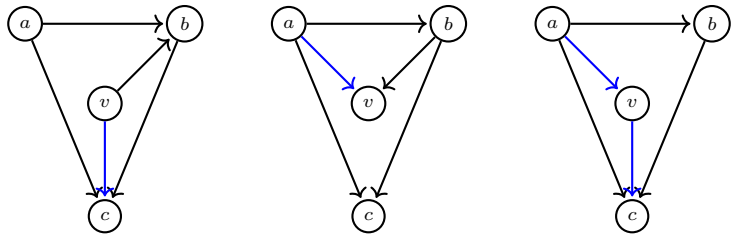
\begin{figure}[h]
{\tiny
    \begin{tikzpicture}[->,shorten >=1pt,auto,node distance=1.5cm,thick,
                                 main node/.style={circle,draw,font=\bfseries}]
      \node[main node] (a) {$a$};
      \node[main node] (v) [below right of=a] {$v$};
      \node[main node] (b) [above right of=v] {$b$};
      \node[main node] (c) [below of=v] {$c$};
      \path
      (a) edge (b) edge (c)
      (b) edge (c)
      (v) edge (b);
      \color{blue} \path (v) edge (c);
\end{tikzpicture}}
\qquad
{\tiny
    \begin{tikzpicture}[->,shorten >=1pt,auto,node distance=1.5cm,thick,
                                 main node/.style={circle,draw,font=\bfseries}]
      \node[main node] (a) {$a$};
      \node[main node] (v) [below right of=a] {$v$};
      \node[main node] (b) [above right of=v] {$b$};
      \node[main node] (c) [below of=v] {$c$};
      \path
      (a) edge (b) edge (c)
      (b) edge (c) edge (v);
      \color{blue} \path (a) edge (v);
\end{tikzpicture}}
\qquad
{\tiny
    \begin{tikzpicture}[->,shorten >=1pt,auto,node distance=1.5cm,thick,
                                 main node/.style={circle,draw,font=\bfseries}]
      \node[main node] (a) {$a$};
      \node[main node] (v) [below right of=a] {$v$};
      \node[main node] (b) [above right of=v] {$b$};
      \node[main node] (c) [below of=v] {$c$};
      \path
      (a) edge (b) edge (c) 
      (b) edge (c);
      \color{blue} \path (a) edge (v) (v) edge (c);
\end{tikzpicture}}
\caption{When $\Gamma'=\Gamma_3$ and $v$ has order two, there are only three graphs $\Gamma$ satisfying (III). Each
  one of them is reducible.\label{fig:dim4_1}}

\end{figure}
%

We now consider the case in which $\Gamma'$ is of type (II). Modulo a reversal of the direction of all arcs, we may
assume that both arcs of $\Gamma'$ are ending in the same vertex $b$. There are again three possible cases: if
there is an arc between $b$ and $v$ it must be from $v$ to $b$, in view of (II), and we may assume by symmetry that
the other arc is between $c$ and $v$; the direction of the arc between $c$ and $v$ is irrelevant, up to
isomorphism. This gives the first case in Figure \ref{fig:dim4_2}. If there is no arc between $b$ and $v$, both
arcs incident with $v$ must either start from $v$ or end in $v$, which leads to the other two cases in Figure
\ref{fig:dim4_2}. The first case is the second graph in Figure \ref{fig:dim4}, while the other two cases are
reducible.

\begin{figure}[h]
{\tiny
    \begin{tikzpicture}[->,shorten >=1pt,auto,node distance=1.5cm,thick,
                                 main node/.style={circle,draw,font=\bfseries}]
      \node[main node] (a) {$a$};
      \node[main node] (b) [right of=a] {$b$};
      \node[main node] (c) [right of=b] {$c$};
      \node[main node] (v) [below of=b] {$v$};
      \path
      (a) edge (b) 
      (c) edge (b)
      (v) edge (b);
      \color{blue} \path (v) edge (c);
\end{tikzpicture}}
\qquad
{\tiny
    \begin{tikzpicture}[->,shorten >=1pt,auto,node distance=1.5cm,thick,
                                 main node/.style={circle,draw,font=\bfseries}]
      \node[main node] (a) {$a$};
      \node[main node] (b) [right of=a] {$b$};
      \node[main node] (c) [right of=b] {$c$};
      \node[main node] (v) [below of=b] {$v$};
      \path
      (a) edge (b) 
      (c) edge (b)
      (v) edge (a);
      \color{blue} \path (v) edge (c);
\end{tikzpicture}}
\qquad
{\tiny
    \begin{tikzpicture}[->,shorten >=1pt,auto,node distance=1.5cm,thick,
                                 main node/.style={circle,draw,font=\bfseries}]
      \node[main node] (a) {$a$};
      \node[main node] (b) [right of=a] {$b$};
      \node[main node] (c) [right of=b] {$c$};
      \node[main node] (v) [below of=b] {$v$};
      \path
      (a) edge (b) edge (v)
      (c) edge (b);
      \color{blue} \path (c) edge (v);
\end{tikzpicture}}
\caption{When $\Gamma'$ is of type (II) there are up to isomorphism and reversal of all arcs only three
  possibilities. The second and third ones are reducible, while the first one is the second graph in Figure
  \ref{fig:dim4}.\label{fig:dim4_2}}
\end{figure}
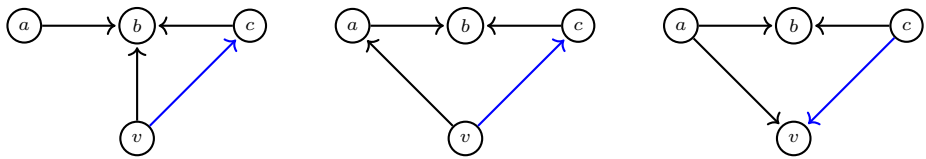
The final case to be considered is when $\Gamma'$ is of type (I). Let us call $c$ the isolated vertex in~$\Gamma'$;
since $\Gamma$ is connected, there must be an arc between $v$ and $c$, which by reversing all arcs of $\Gamma$
may be assumed to be from $c$ to $v$. Let the labeling of the other vertices be such that the unique arc in
$\Gamma'$ is from $a$ to $b$. Then, in view of (II), the third arc of $\Gamma$ must be from $a$ to~$v$. Then
$\Gamma$ is the third graph in Figure \ref{fig:dim4}.

We have now found all possible connected irreducible 4-vertex graphs for which every 3-vertex induced subgraph has
the \KP. To show that $\Ga4$, which is the first graph in Figure \ref{fig:dim4}, is the only one having the \KP,
we need to show that none of the other two graphs in Figure \ref{fig:dim4} has the \KP. We do this by using the
method described in Section~\ref{par:compute_poisson}. We do this only for one of the graphs, as both graphs are
very similar, and hence also the computations to be done.

We consider the last graph in Figure \ref{fig:dim4} and label the vertices as follows:
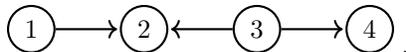
\begin{figure}[h]
  {\small
    \begin{tikzpicture}[scale=2,->,shorten >=1pt,auto,thick,node distance=1.5cm,
                                 main node/.style={circle,draw,font=\bfseries}]
      \node[main node] (1) {$1$};
      \node[main node] (2) [right of=1] {$2$};
      \node[main node] (3) [right of=2] {$3$};
      \node[main node] (4) [right of=3] {$4$};
      \path
      (1) edge (2)
      (3) edge (2) edge (4); 
\end{tikzpicture}\;.}
\end{figure}

\noindent We choose the point $P:=(-1,1,-1,1)$. Then $M(P)$ and its inverse are given by
\begin{equation*}
  M(P)=  \begin{pmatrix}
    0&1&0&0\\
    1&-1&1&0\\
    0&1&-1 &1\\
    0&0&1&0
  \end{pmatrix},\quad
  M(P)^{-1}=
  \begin{pmatrix}
    1&1&0&-1\\
    1&0&0&0\\
    0&0&0&1\\
    -1&0&1&1
  \end{pmatrix}\;.
\end{equation*}
It follows that  $\tP=(-1,-1,1,1)$. Evaluating $\tx 1,\dots,\tx 4$ at $P$ we get 
\begin{equation*}
  \cL(\tP)=\begin{pmatrix}
    1-x_2\\
    1+x_1+x_2+x_3\\
    -1+x_2+x_3+x_4\\
    -1-x_3
  \end{pmatrix}\;,
\end{equation*}%
so that
\begin{equation*}
  \pb{\tx1,x_1}(P)=\pb{\tx1,x_3}(P)=0\;,  \pb{\tx1,x_2}(P)=-3\;, \pb{\tx1,x_4}(P)=-2\;.
\end{equation*}%
From these values we find that $\pb{\tx1,\tx2}(P)=3$, so that $\pb{\tx1,\tx2}(P)-\tx1(P)\tx2(P)=2$, which proves
that the last graph in Figure \ref{fig:dim4} does not have the \KP.

\subsection{The higher-dimensional case}\label{par:dimn}
We are now ready to prove Theorem~\ref{thm:main} for graphs with more than four vertices.

We proceed by
contradiction: we assume that $\Gamma$ is a graph with $n>4$ vertices which is not a tournament, is connected,
irreducible and has the \KP, and show that this leads to a contradiction. First notice that~$\Gamma$ has at least
two vertices at distance 2; this is easily seen by considering a shortest chain between any pair of non-adjacent
vertices: the first and third vertex of such a chain are at distance 2. Let $s$ and $t$ be two vertices at distance
2 and let $u$ be any vertex adjacent to both $s$ and $t$. According to~(II), the arcs between $u$ on the one hand
and $s$ and $t$ on the other hand must both be starting from $u$ or ending in $u$. By reversing all arcs if
needed, we may assume that the subgraph of $\Gamma$ induced by the vertices $s,t$ and $u$ is given by
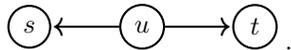
\begin{figure}[h]
  {\small
    \begin{tikzpicture}[scale=2,->,shorten >=1pt,auto,node distance=1.5cm,thick,
                                 main node/.style={circle,draw,font=\bfseries}]
      \node[main node] (s) {$s$};
      \node[main node] (u) [right of=s] {$u$};
      \node[main node] (t) [right of=u] {$t$};
      \path
      (u) edge (s) edge (t); 
    \end{tikzpicture}\;.}
\end{figure}

Since $\Gamma$ is irreducible, $s$ and $t$ cannot have the same neighbors. By the symmetry in $s$ and $t$, we may
suppose that $t$ has an (in or out) neighbor, say $v$, which is not an (in or out) neighbor of $s$. In fact, in
either case, $v$ cannot be adjacent to $s$ because of~(II), applied to the subgraph induced by $\set{s,t,v}$.
This leads to three cases, depending on the direction of the arc between $v$ and $t$ and on whether or not there is
an arc between $u$ and $v$; notice that if there is an arc between $u$ and $v$, it must go from $u$ to $v$, again
because of (II), applied to the subgraph induced by $\set{s,u,v}$. The three cases are displayed in Figure
\ref{fig:dim4_3}.

\begin{figure}[h]
{\tiny
    \begin{tikzpicture}[->,shorten >=1pt,auto,node distance=1.5cm,thick,
                                 main node/.style={circle,draw,font=\bfseries}]
      \node[main node] (v) {$v$};
      \node[main node] (s) [below left of=v] {$s$};
      \node[main node] (t) [below right of=v] {$t$};
      \node[main node] (u) [below left of=t] {$u$};
      \path
      (u) edge (s) edge (t) 
      (v) edge (t);
\end{tikzpicture}}
\qquad
{\tiny
    \begin{tikzpicture}[->,shorten >=1pt,auto,node distance=1.5cm,thick,
                                 main node/.style={circle,draw,font=\bfseries}]
      \node[main node] (v) {$v$};
      \node[main node] (s) [below left of=v] {$s$};
      \node[main node] (t) [below right of=v] {$t$};
      \node[main node] (u) [below left of=t] {$u$};
      \path
      (u) edge (s) edge (t) edge (v)
      (v) edge (t);
\end{tikzpicture}}
\qquad
{\tiny
    \begin{tikzpicture}[->,shorten >=1pt,auto,node distance=1.5cm,thick,
                                 main node/.style={circle,draw,font=\bfseries}]
      \node[main node] (v) {$v$};
      \node[main node] (s) [below left of=v] {$s$};
      \node[main node] (t) [below right of=v] {$t$};
      \node[main node] (u) [below left of=t] {$u$};
      \path
      (u) edge (s) edge (t) edge (v)
      (t) edge (v);
\end{tikzpicture}}
\caption{Under the above assumptions, $\Gamma$ must contain an induced subgraph of the above
  type. \label{fig:dim4_3}}
\end{figure}
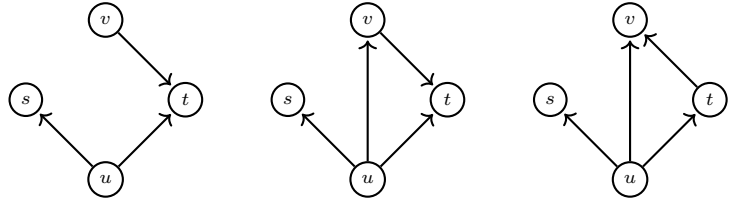
The first graph in Figure \ref{fig:dim4_3} is isomorphic to the third graph in Figure \ref{fig:dim4}, while the
other two are isomorphic to the second graph in that figure. We have shown that these graphs do not have the \KP,
leading to a contradiction. We have therefore proven Theorem~\ref{thm:main} for all $n$.

\section{Deformed Lotka-Volterra systems}\label{sec:LV_deformed}
In this section we consider deformations of Lotka-Volterra systems, associated with diagonal Poisson brackets which
are deformed by constants. We first introduce and study these deformed Poisson brackets, which are associated with
augmented graphs, and define the \KP\ for such graphs. We then characterize these graphs showing that they are
closely related to the graphs $\Ga n$ and their clonings.

\subsection{Deformed diagonal Poisson brackets}

We consider in this section deformations of diagonal Poisson brackets with constant terms; they lead to
Hamiltonian systems which are deformations of Lotka-Volterra systems, and which we will study in the next
subsection.
\begin{defn}
  A Poisson bracket $\PBb$ on $\bbF(x)=\bbF(x_1,\dots,x_n)$ is said to be a \emph{diagonal Poisson bracket,
    deformed by constants}, or simply, a \emph{deformed diagonal Poisson bracket}, if for all $1\leqslant
  i,j\leqslant n$, the Poisson brackets $\pbb{x_i,x_j}$ are of the form $\pbb{x_i,x_j}=a_{i,j}x_ix_j+b_{i,j}$,
  where $a_{i,j},b_{i,j}\in\bbF$.  The matrix $B=(b_{i,j})$ is called a \emph{deformation matrix of $A$ (or of
    $\PB$)}.
\end{defn}
The conditions on a matrix $B$ for it to be a deformation matrix of a diagonal Poisson structure are given in the
following proposition.
\begin{prop}\label{prp:comp}
  Let $A=(a_{i,j})$ and $B=(b_{i,j})$ be two $n\times n$ skew-symmetric matrices. Then~$B$ is a deformation matrix
  of $A$ if and only if
  \begin{equation}\label{eq:comp_cond}
    b_{i,j}(a_{i,k}+a_{j,k})=0\;,\quad \text{for all} \quad i,j,k\in S \quad \text{with} \quad i\neq j\neq k \neq i\,.
  \end{equation}
\end{prop}
\begin{proof}
We define a biderivation $\PBb$ of $\bbF(x)$, by $\pbb{x_i,x_j}:=a_{i,j}x_ix_j+b_{i,j}$ for $1\leqslant
i,j\leqslant n$.  By definition, $B$ is a deformation matrix of $A$ if and only if $\PBb$ satifies the Jacobi
identity, i.e.,
$$
  \pbb{\pbb{x_i,x_j},x_k}+\cycl(i,j,k)=0, \quad \text{for all} \quad 1\leqs i<j<k\leqs n\;,
$$
where $\cycl(i,j,k)$ means cyclic permutation of the indices $i, j, k$. Using the fact that $\PBb$ is a
biderivation and that $\PB$ is a Poisson bracket, this is equivalent to
\begin{align*}
  0&=a_{i,j}\pbb{x_i x_j,x_k}+\cycl(i,j,k)=a_{i,j}b_{j,k}x_i+a_{i,j}b_{i,k}x_j+\cycl(i,j,k)\\
  &=a_{k,i}b_{i,j}x_k+a_{j,k}b_{j,i}x_k+\cycl(i,j,k)=b_{i,j}(a_{k,i}+a_{k,j})x_k+\cycl(i,j,k)\,,
\end{align*}
which in turn amounts to the condition that $b_{i,j}(a_{i,k}+a_{j,k})=0$ whenever the indices $i,j,k$ are
different, which is precisely \eqref{eq:comp_cond}.
\end{proof}
Condition \eqref{eq:comp_cond} can be stated equivalently by saying that the $(i,j)$-th entry $b_{i,j}$ of $B$ is
zero whenever there exists an index $k$, different from $i$ and~$j$, with $a_{i,k}+a_{j,k}\neq0$. It follows that,
given $A$, there are two types of pairs of distinct indices $(i,j)$, depending on whether or not
$a_{i,k}+a_{j,k}=0$ for all~$k$, different from $i$ and~$j$; in the positive case, one can assign any value
to~$b_{i,j}$, while that value must be zero in the negative case, for the constants~$b_{i,j}$ to define a
deformation matrix of $A$.

The condition that $B$ is a deformation matrix of $A$ can easily be read off from the skew-symmetric graph
$\Gamma=(S,A)$, associated to $A$.  Given three different vertices $i,j,k\in S$, we say that $k$ is an
\emph{opposite neighbor} of $i$ and $j$ if the arcs from $i$ to $k$ and from $j$ to $k$ have opposite values,
$a_{i,k}+a_{j,k}=0$; the vertices $i$ and $j$ are said to have \emph{opposite~neighborhoods} if every other vertex
$k$ is an opposite neighbor of $i$ and $j$.  On a picture, representing $\Gamma$, we will add a green arc from $i$ to
$j$ when $i$ and $j$ have opposite neighborhoods and, say, $i<j$. By the above, this indicates that if one puts any
values at the positions in the $B$-matrix which correspond to green arcs, and zeros at all other positions, then $B$
is a deformation matrix of $A$ and all deformation matrices of $A$ are obtained in this way. We call any triplet
$\Delta=(S,A,B)$ with $B$ a deformation matrix of $A$ an \emph{augmentation of $\Gamma$} and refer to $\Delta$ as
an \emph{augmented graph}. A green arc from $i$ to~$j$ may be labeled with the value $b_{i,j}$ but that will not be
needed in what follows.

\begin{example}
Recall that for the graph $\Ga n$ there is an arc from $i$ to $j$ if and only if $i<j$. The vertices $1$ and $n$
have opposite neighborhoods and are the only vertices with this property. Therefore, the only possible non-zero
entries of the deformation matrix $B$ are $b_{1,n}=-b_{n,1}$. See the first picture in Figure
\ref{fig:augmented_graphs} below for the case of $n=6$.
\end{example}

\begin{example}\label{exa:delta_n}
We denote by $C_n$ the graph with $n$ vertices $S=\set{1,2,\dots,n}$ and an arc of value~1 from $i$ to $i+1 \mod n$
for $i\in S$. When $n=3$, any two vertices have opposite neighborhoods; when $n=4$, the vertices $1$ and~$3$ have
opposite neighborhoods, as well as the vertices $2$ and $4$; when $n>4$ no two vertices have opposite
neighborhoods.  See the second picture in Figure \ref{fig:augmented_graphs} for the case of $n=4$.
\end{example}

\begin{figure}[h]
\begin{tikzpicture}[->,shorten >=1pt,auto,node distance=2cm,
                thick,main node/.style={circle,draw,font=\bfseries}]
  \node[main node] (1) {$1$};
  \node[main node] (2) [right of=1] {$2$};
  \node[main node] (3) [below right of=2] {$3$};
  \node[main node] (4) [below left of=3] {$4$};
  \node[main node] (5) [left of=4] {$5$};
  \node[main node] (6) [below left of=1] {$6$};
  \path
    (1) edge (2) edge(3) edge (4) edge (5) edge (6)
    (2) edge (3) edge (4) edge (5) edge (6)
    (3) edge (4) edge (5) edge (6)
    (4) edge (5) edge (6)
    (5) edge (6)
  (1) edge[green,bend right] (6) ;
\end{tikzpicture}
\qquad\qquad
\begin{tikzpicture}[->,shorten >=1pt,auto,node distance=2cm,
                thick,main node/.style={circle,draw,font=\bfseries}]
  \node[main node] (1) {$1$};
  \node[main node] (2) [right of=1] {$2$};
  \node[main node] (3) [below of=2] {$3$};
  \node[main node] (4) [below of=1] {$4$};
  \path
    (1) edge (2)
    (2) edge (3)
    (3) edge (4)
    (4) edge (1)
    (1) edge[green] (3)
    (2) edge[green] (4);
\end{tikzpicture}
\caption{Augmented graphs of $\Ga 6$ and $C_4$. The green arcs join vertices with opposite
  neighborhoods.\label{fig:augmented_graphs}}
\end{figure}
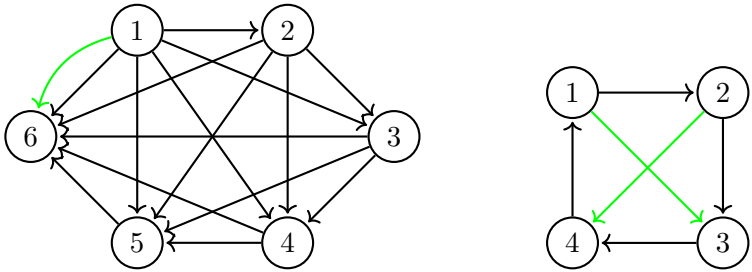
We show in the next lemma how the augmentations of a graph and of its clonings are related.
\begin{lemma}\label{lma:augmented}
  Let $\Gamma$ be a skew-symmetric graph and let $\wght$ be a weight vector for $\Gamma$. Two vertices $(i,k)$
  and $(j,\ell)$ of the cloned graph $\Gamma^\wght$ have opposite neighborhoods if and only if the following two
  conditions are fulfilled:
  \begin{enumerate}
    \item The vertices $i$ and $j$ have opposite neighborhoods in $\Gamma$;
    \item If there is an arc between $i$ and $j$ in $\Gamma$, then $\wght(i)=\wght(j)=1$.
  \end{enumerate}
\end{lemma}

\begin{proof}
As before, we let $\Gamma=(S,A)$ with $A=(a_{i,j})$. Let $i,j$ be two vertices of $\Gamma$ and let us assume that
the above conditions (1) and (2) are satisfied. Let $(r,s)$ be a vertex of $\Gamma^\wght$, which is different from
some given distinct vertices $(i,k)$ and $(j,\ell)$ of $\Gamma^\wght$. Then
\begin{align*}
  &\hbox{There is an arc  $(i,k)\to(r,s)$ in $\Gamma^\wght$ with value $a\neq0$}\\
  &\qquad\Longleftrightarrow\hbox{There is an arc  $i\to r$ in $\Gamma$ with value $a\neq0$}\\
  &\qquad\Longleftrightarrow\hbox{There is an arc  $j\to r$ in $\Gamma$ with value $-a\neq0$}\\
  &\qquad\Longleftrightarrow\hbox{There is an arc  $(j,\ell)\to(r,s)$ in $\Gamma^\wght$ with value $-a\neq0$}\;,
\end{align*}
which means that $(r,s)$ is an opposite neighbor to $(i,k)$ and $(j,\ell)$; this shows that $(i,k)$ and $(j,\ell)$
have opposite neighborhoods. The second equivalence is a direct consequence of (1) when $r\neq i,j$, but needs some
explanation when $r=i$ or $r=j$: it is clearly valid when there is no arc between $i$ and $j$, but when there is an
arc between $i$ and $j$ then according to (2), $\wght(i)=\wght(j)=1$, so that in fact $r\neq i$ and $r\neq j$. This
shows that the conditions (1) and (2) are sufficient.

We now show that these conditions are also necessary. If (1) does not hold, then there exists a vertex $r$ of
$\Gamma$ which is not an opposite neighbor of $i$ and $j$. Then for any $k,\ell$, the vertex $(r,1)$ is not an
opposite neighbor of $(i,k)$ and $(j,\ell)$, so that $(i,k)$ and $(j,\ell)$ do not have opposite neighborhoods. If
(2) does not hold, there is an arc from $i$ to $j$, with value $a\neq0$, but $\wght(i)\cdot\wght(j)>1$, say
$\wght(i)>1$. Let $(i,m)$ be a vertex of $\Gamma^\wght$, with $k\neq m$. Then there is an arc from
$(i,m)$ to $(j,\ell)$ with value $a\neq0$, but there is no arc between $(i,k)$ and $(i,m)$. It follows
again that $(i,k)$ and $(j,\ell)$ do not have opposite neighborhoods.
\end{proof}

\begin{example}\label{ex:Gamma_n_clon_def}
The previous lemma, applied to the weighted graph $(\Gamma_n,\wght)$, shows that $\Gamma_n^\wght$ can only have a
pair of vertices with opposite neighborhoods when $\wght(1)=\wght(n)=1$; in this case, $(1,1)$ and $(n,1)$ have
opposite neighborhoods and is the only pair of vertices with this property. See Figure~\ref{fig:cloned_Gamma_n}.
\end{example}

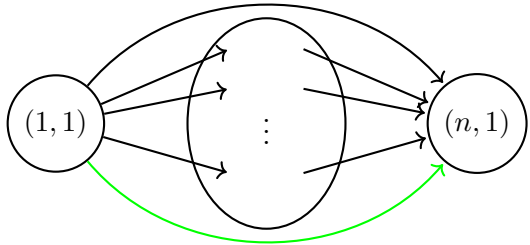
\begin{figure}[h]
\begin{tikzpicture}[->,shorten >=1pt,auto,node distance=2cm,
                thick,main node/.style={circle,draw,font=\bfseries},scale=0.7]
  \node[main node] (1) {$(1,1)$};
  \node[main node] (2) at (8,0) {$(n,1)$};
  \node (3) at (3.5,1.5) {};
  \node (4) at (3.5,.7) {};
  \node (5) at (4,0) {$\vdots$};
  \node (6) at (3.5,-1) {};

  \node (7) at (4.5,1.5) {};
  \node (8) at (4.5,.7) {};
  \node (10) at (4.5,-1) {};

  \draw (4,0) ellipse (1.5cm and 2cm);

  \path
    (1)  edge[bend left=50] (2)
    (1) edge[green,bend right=50] (2)
    (1)  edge (3) edge (4) edge (6)
    (7)  edge (2)
    (8)  edge (2)
    (10)  edge (2);
\end{tikzpicture}
 \caption{When $\wght(1)=\wght(n)=1$ the vertices $(1,1)$ and $(n,1)$ of $\Gamma_n^\wght$ are the only vertices with
   opposite neighborhoods.\label{fig:cloned_Gamma_n}}
\end{figure}

\subsection{The Kahan-Poisson property for deformations of Lotka-Volterra systems}\label{subsec:LV_defomations}
We show in this subsection that the deformations with constant terms of the Lotka-Volterra systems~$\LV(\Gamma)$
for which their Kahan map is a Poisson map with respect to the corresponding deformed Poisson structure, are
precisely those for which their underlying graph $\Gamma$ has the \KP. According to Theorem~\ref{thm:main}, this
means that every connected component of $\Gamma$ is isomorphic to $\ga\Ga n$ for some $\ga\in\bbF^*$ and
$n\in\bbN^*$.

We first recall from \cite{kahan} the recipe of the Kahan discretization of a general class of systems of
differential equations which covers the deformations of Lotka-Volterra systems that we consider. To this end,
consider a system of differential equations on $\bbF^n$ of the form
\begin{equation}\label{eq:def_sys}
  \dot{x}_i=Q_i(x)+c_i\;, \quad i=1,2,\dots, n\,,
\end{equation}
where $x=(x_1, x_2, \dots, x_n),\, Q_i$ is a quadratic form and $c_i\in\bbF$. If we denote by $B_i$ the bilinear
form, corresponding to $Q_i$, so that $Q_i(x)=B_i(x,x)$, then the \emph{Kahan discretization of \eqref{eq:def_sys}}
is given by
\begin{equation}\label{eq:def_kah_impl}
  \frac{\tx i-x_i}{2\varepsilon}=B_i(\tx,x)+c_i\;, \quad i=1,2,\dots,n\,.
\end{equation}
Solving \eqref{eq:def_kah_impl} linearly for $\tx 1, \tx 2, \dots, \tx n$ we get a family of birational maps on
$\bbF^n$, parametrized by the step size $\varepsilon$. As in the undeformed case, for fixed $\varepsilon$ we call
the corresponding map the \emph{Kahan map}.  The corresponding endomorphism of the field of rational functions
$\bbF(x)$, defined for $i=1,\dots,n$ by $\Kb(x_i):=\tx i$, is called the \emph{Kahan morphism}.

We now introduce the deformed Lotka-Volterra systems which we will study. Let $\Delta=(S,A,B)$ be an augmented
graph, whose associated Poisson bracket on $\bbF(x)$ is denoted by $\PBb$. Recall that it is defined by
$\pbb{x_i,x_j}:=a_{i,j}x_ix_j+b_{i,j}$ for $i,j \in S$. Taking $H:=x_1+x_2+\cdots+x_n$ as Hamiltonian, the
Hamiltonian vector field $\Pbb{H}$ is given by the following system of differential equations:
\begin{equation}\label{eq:def_LV_syst}
\dot x_i = \sum_{j=1}^n a_{i,j} x_i x_j+c_i\;, \ \ i\in S\,,
\end{equation}
where the parameters $c_i$ are related with $B=(b_{i,j})$ by $c_i=\sum_{j=1}^n b_{i,j}$. It is a deformation of
the Lotka-Volterra system \eqref{eq:LV_gen_intro} and is of the above form~\eqref{eq:def_sys}.
Using~\eqref{eq:def_kah_impl}, the Kahan map of \eqref{eq:def_LV_syst} is implicitly given by
\begin{equation}\label{eq:def_LV_kahan}
\frac{\tx i-x_i}{\varepsilon}=\tx i\sum_{j=1}^na_{i,j}x_j+x_i\sum_{j=1}^na_{i,j}\tx j+2c_i\,, \quad i\in S\,.
\end{equation}

\begin{defn}\label{def:Kah_pois_prop_def}
  An augmented graph $\Gb$ is said to have the \emph{\KP}\ if the Kahan map \eqref{eq:def_LV_kahan} with step size
  $\varepsilon=1$ is a Poisson map with respect to the Poisson bracket $\PBb$.
\end{defn}
We first prove an analog of Proposition \ref{prp:kahan_cloning} for augmented graphs. Let $\Gamma=(S,A)$ be a
connected skew-symmetric graph with vertex set $S=\set{1,2,\dots,n}$ and let $\wght$ be a weight vector for
$\Gamma$. Let $\Delta^\wght=(S^\wght,A^\wght,B^\wght)$ be an augmented graph of $\Gamma^\wght$ and let
$\Delta=(S,A,B)$ be the augmented graph of $\Gamma$, defined for $1\leqs i,j\leqs n$ by
$b_{i,j}:=\sum_{k=1}^{\wght(i)}\sum_{\ell=1}^{\wght(j)}b^\wght_{(i,k),(j,\ell)}$; notice that, according to Lemma
\ref{lma:augmented}, this defines indeed an augmented graph of $\Gamma$. The Poisson brackets on $\bbF(x)$ and
$\bbF(y)$) associated with $\Delta$ and $\Delta^\wght$ are denoted by $\PB_b$ and $\PB_b^\wght$ respectively. The
Kahan maps on $\bbF(x)$ and $\bbF(y)$ with $\varepsilon=1$ are denoted by $\Kb$ and $\Kb^\wght$; for $u\in\bbF(x)$
and $v\in\bbF(y)$ we also write $\tilde u$ for $\Kb(u)$ and $\bar v$ for $\Kb^\wght(v)$. The decloning map is, as
before, denoted by $\Sigma$.
\begin{prop}\label{prp:kahan_cloning_defo}
  With the above definitions and notations, the following diagram of fields and field morphisms is commutative:
\begin{equation}\label{dia:kahan_diagram_2}
  \begin{tikzcd}[row sep=5ex, column sep=7ex]
    (\bbF(x),\PB_b)\arrow{r} {\Kb}\arrow[swap]{d} {\Sigma}&(\bbF(x),\PB_b)\arrow{d}{\Sigma}\\
          {(\bbF(y),\PB_b^\wght)}\arrow[swap]{r}{\Kb^\wght}&{(\bbF(y),\PB_b^\wght)}
  \end{tikzcd}
\end{equation}
  The vertical arrow $\Sigma$ is a Poisson morphism. If $K^\wght_b$ is a Poisson morphism, then $K_b$ also;
  if~$\Gamma$ is isomorphic to $\gamma\Ga n$ for some $n\in\bbN^*$ and some $\gamma\in \bbF^*$, so that $K_b$ is a
  Poisson morphism, then $\Kb^\wght$ is also a Poisson morphism.
\end{prop}

\begin{proof}
The commutativity of the diagram is shown in exactly the same way as in the proof of
Proposition~\ref{prp:kahan_cloning}. For $1\leqslant i<j\leqslant n$,
we have that 
\begin{align*}
  \pb{\Sigma(x_i),\Sigma(x_j)}_b^\wght
  &=\sum_{k=1}^{\wght(i)}\sum_{\ell=1}^{\wght(j)}\pb{y_{i,k},y_{j,\ell}}^\wght_b
  =\sum_{k=1}^{\wght(i)}\sum_{\ell=1}^{\wght(j)}\(a_{i,j}y_{i,k}y_{j,\ell}+b_{(i,k),(j,\ell)}\)\\
  &=a_{i,j}\Sigma(x_i)\Sigma(x_j)+b_{i,j}=\Sigma\pb{x_i,x_j}_b\;,
\end{align*}
which shows that $\Sigma$ is a Poisson morphism. Using the commutativity of \eqref{dia:kahan_diagram_2} and the
fact that~$\Sigma$ is a morphism of Poisson fields, we get
$$
  \Sigma\widetilde{\pb{x_i,x_j}_b}=\overline{\Sigma\pb{x_i,x_j}_b}=\overline{\pb{\Sigma(x_i),\Sigma(x_j)}_b^\wght}
$$
and 
$$
  \Sigma\pb{\tx i,\tx j}_b=\pb{\Sigma(\tx i),\Sigma(\tx j)}^\wght_b=
  \pb{\overline{\Sigma(x_i)},\overline{\Sigma(x_j)}}^\wght_b\;.
$$
It follows that, if $K^\wght_b$ is a Poisson morphism, then $\Sigma\widetilde{\pb{x_i,x_j}_b}=\Sigma\pb{\tx i,\tx
j}_b,$ which implies, by the injectivity of $\Sigma$, that $\widetilde{\pb{x_i,x_j}_b}=\pb{\tx i,\tx j}_b$. We have
hereby shown that if $K^\wght_b$ is a Poisson morphism, then $K_b$ also. 

Suppose now that $\Gamma=\gamma\Ga n$ where $n\in\bbN^*$ and $\gamma\in \bbF^*$. It was shown in
\cite[Theorem~5.8]{PPP_Red_DC} that the Kahan map of $\Ga n$ is a Poisson map, so that $\Kb$ is a Poisson
morphism. We show that $\Kb^\wght$ is also a Poisson morphism. According to Example \ref{ex:Gamma_n_clon_def}, if
$\wght(1)\cdot\wght(n)>1$ then the only deformation matrix of $\gamma\Ga n^\wght$ is the zero matrix, and there is
nothing to prove. We therefore suppose that $\wght(1)=\wght(n)=1$ and denote
$b:=b^\wght_{(1,1),(n,1)}=-b^\wght_{(n,1),(1,1)}$, which is the only entry of $B^\wght$ which is possibly
non-zero. We show as in the proof of Proposition~\ref{prp:kahan_cloning} that, for any $i\in S$ and for any
$1\leqslant k\leqslant \wght(i)$,
\begin{equation}\label{eq:def_Casimirs}
  \Sigma(x_i)/y_{i,k} \hbox{ is a Casimir of }\PB^\wght_b \hbox{ and an invariant of } \Kb^\wght\;.
\end{equation}%
For $i=1$ or $i=n$ this is obvious because then $k=1$ and $\Sigma(x_1)=y_{1,1}$ and $\Sigma(x_n)=y_{n,1}.$ When
$1<i<n$ we have that $\pB{y_{i,k}}^\wght_b=\pB{y_{i,k}}^\wght$ and so $y_{i,m}/y_{i,k}$ and $\Sigma(x_i)/y_{i,k}$
are Casimirs of $\PB^\wght_b$ as well. The fact that $y_{i,m}/y_{i,k}$, and hence $\Sigma(x_i)/y_{i,k}$, is an
invariant of $\Kb^\wght$ for $1<i<n$ is shown in exactly the same way as in the proof of
Proposition~\ref{prp:kahan_cloning}. For $1\leqs i<j\leqslant n$, $1\leqs k\leqs\wght(i)$ and
$1\leqs\ell\leqs\wght(j)$ the computation of the Poisson brackets $\pb{\ty ik,\ty j\ell}^\wght_b$ can therefore be
done as in \eqref{equ:Kahan_Poisson}, yielding formally the same result, to wit
\begin{align}\label{equ:Kahan_Poisson_defo}
  \pb{\ty ik,\ty j\ell}^\wght_b=\bar y_{i,k}\bar y_{j,\ell}\Sigma\frac{\pb{\tx i,\tx j}_b}{\tx i\tx j}\;.
\end{align}
When $i=1$ and $j=n$, $\pb{\tx i,\tx j}_b=\pb{\tx 1,\tx n}_b=a_{1,n}\tx1\tx n+b$ and $\bar y_{i,k}=\bar
y_{1,1}=\Sigma(x_1)$ and $\bar y_{j,\ell}=y_{n,1}=\Sigma(x_n)$ since $\wght(1)=\wght(n)=1$, so that
\eqref{equ:Kahan_Poisson_defo} becomes
\begin{equation*}
  \pb{\ty 11,\ty n1}^\wght_b=a_{1,n}\ty11\ty n1+b=\overline{\pb{y_{1,1},y_{n,1}}^\wght_b}\;,
\end{equation*}
as wanted. Otherwise, $\pb{\tx i,\tx j}_b=a_{i,j}\tx i\tx j$, so that \eqref{equ:Kahan_Poisson_defo}
becomes
\begin{equation*}
  \pb{\ty ik,\ty j\ell}^\wght_b=a_{i,j}\ty ik\ty j\ell=\overline{\pb{y_{i,k},y_{j,\ell}}^\wght_b}\;,
\end{equation*}
which finishes the proof that $\Kb^\wght$ is a Poisson morphism.
\end{proof}

We use the above proposition and Theorem \ref{thm:main} to show that the \KP\ is preserved under deformation.
\begin{prop}\label{prp:kahan_deformed}
  Let $\Gamma$ be a connected skew-symmetric graph. Then the following statements are equivalent:\goodbreak
    \begin{enumerate}
      \item [(i)] $\Gamma$ has the \KP; 
      \item [(ii)] All augmented graphs $\Gb$ of $\Gamma$ have the \KP;
      \item [(iii)] Some augmented graph $\Gb$ of $\Gamma$ has the \KP.
    \end{enumerate}
\end{prop}
\begin{proof}
We first prove that (i) implies (ii). If $\Gamma$ has the \KP, then we know from Theorem~\ref{thm:main} that
$\Gamma$ is isomorphic to $\gamma\Ga n^\wght$ for some $n\in\bbN^*$, some $\gamma\in \bbF^*$ and some weight vector
on~$\Ga n$. Let $\Gb$ be any augmented graph of $\Gamma$. According to Proposition \ref{prp:kahan_cloning_defo},
the Kahan morphism~$\Kb^\wght$ is also a Poisson morphism. This shows that $\Gb$ also has the \KP. The proof that
(ii) implies~(iii) is trivial, because any graph can be considered as an augmented graph of itself with the zero
deformation matrix.  Suppose now that $\Gb=(S,A,B)$ is an augmented graph of a skew-symmetric graph $\Gamma$ and
that $\Gb$ has the \KP. Let $\Kb$ be the Kahan morphism of $\Gb$ (with $\varepsilon=1$), which is a Poisson
morphism.  It is clear from Equations \eqref{eq:def_LV_kahan} and \eqref{eq:LV_Kahan} that by setting $b_{i,j}:=0$
for all $i,j\in S$ in~$\Kb$, we get the Kahan morphism $K$ of $\Gamma$ (with $\varepsilon=1$). Since~$\Kb$ is a
Poisson morphism,
$$
  \pbb{\Kb(x_i),\Kb(x_j)}=a_{i,j}\Kb(x_i)\Kb(x_j)+b_{i,j}\;, \quad \text{for all} \quad i,j\in S\;.
$$
Setting $b_{i,j}=0$, the right hand side of the above becomes $a_{i,j}K(x_i)K(x_j)$ while the left hand side becomes
$$
  \pbb{\Kb(x_i),\Kb(x_j)}\bigm|_{b=0}=\pb{\Kb(x_i),\Kb(x_j)}\bigm|_{b=0}=\pb{K(x_i),K(x_j)}\,.
$$
Therefore $\pb{K(x_i),K(x_j)}=a_{i,j}K(x_i)K(x_j)$ for all $i,j\in S$, which means that $\Gamma$ has the \KP\ as
well.
\end{proof}
In combination with Theorem~\ref{thm:main}, Proposition \ref{prp:kahan_deformed} leads at once to Theorem
\ref{thm:main_deformed}.

It also follows from the proposition (or from the theorem) by a simple rescaling argument that an augmented graph
has the \KP\ if and only if the Kahan map of an augmented graph is a Poisson map for some value of
$\varepsilon\in\bbF^*$. Indeed, when $\varepsilon$ is given weight~$-1$, while giving a weight $1$ to all $x_i$ and
a weight $2$ to the parameters $c_i$, the defining equations \eqref{eq:def_LV_kahan} of the Kahan map become
homogeneous, and homotheties of quadratic Poisson structures are Poisson maps, as we already recalled. In
particular, the Kahan map with step size $\varepsilon$ of an augmented graph $\Delta=(S,A,B)$ is a Poisson map, if
and only if the Kahan map with step size $1$ of $(S,A,\varepsilon^2B)$ is a Poisson map, i.e.,
$(S,A,\varepsilon^2B)$ has the KP property; in view of Proposition~\ref{prp:kahan_deformed}, this is equivalent to
$(S,A,B)$ having the \KP.

Another consequence of the proposition is that the final statement in Proposition \ref{prp:kahan_cloning_defo} can
be reformulated as an if and only if, so that Proposition \ref{prp:kahan_cloning_defo} is a generalization of
Proposition~\ref{prp:kahan_cloning}. Indeed, when $\Kb$ is a Poisson morphism, so that $\Delta$ has the \KP, then
$\Gamma$ has the \KP\ (by Proposition \ref{prp:kahan_deformed}), hence also $\Gamma^\wght$ (by
Proposition~\ref{prp:kahan_cloning}), and hence also $\Delta^\wght$ (again by Proposition
\ref{prp:kahan_deformed}), so that $\Kb^\wght$ is a Poisson morphism. It follows that, in the notations of
Proposition \ref{prp:kahan_cloning_defo}, $\Kb$ is a Poisson morphism if and only $\Kb^\wght$ is a Poisson
morphism.

\section{Liouville and superintegrability}\label{sec:integ}

We have shown in the previous sections that the only connected skew-symmetric graphs $\Gamma$ which have
the~\KP\ are of the form $\Gamma=\gamma\Ga n^\wght$, where $\gamma\in\bbF^*$, $n\in\bbN^*$ and $\wght$ is a weight
vector on $\Ga n$; also, that the only augmented graphs $\Delta$ which have the \KP\ are augmented graphs of
$\Gamma=\gamma\Ga n^\wght$. We now show that the corresponding Lotka-Volterra systems $\LV(\Gamma)$ and deformed
Lotka-Volterra systems $\LV(\Delta)$ are both Liouville integrable and superintegrable, and that their Kahan
discretizations are both Liouville integrable and superintegrable as well.  By a simple rescaling argument, already
used several times above, we may assume that $\ga=1$, so we will consider in what follows only $\Ga n^\wght$ and
its augmentations.

We first fix the notation and the context. Let $n\in\bbN^*$, let $\Ga n^\wght$ be a cloned graph of~$\Ga n$, with
$\wght(1)=\wght(n)=1$, and let $b\in\bbF$. We denote by $\Delta_n$ the augmented graph of~$\Ga n$, where the
deformation matrix $B$ has as only possible non-zero entries $b_{1,n}=-b_{n,1}:=b$. Similarly,~$\Delta_n^\wght$
denotes the augmented graph of $\Ga n^\wght$, where the deformation matrix $B$ has as only possible non-zero
entries $b_{(1,1),(n,1)}=-b_{(n,1),(1,1)}:=b$.  We consider the fields $\bbF(x)$ and $\bbF(y)$, where
$x=(x_1,\dots,x_n)$ and $y=(y_{1,1},y_{2,1},y_{2,2},\dots, y_{n,1})$, as before.
We consider on $\bbF(x)$ the Poisson bracket $\PB_b$, associated with $\Delta_n$, and on $\bbF(y)$ the Poisson
bracket $\PB_b^\wght$, associated with~$\Delta_n^\wght$.


Since we will take in this section a more geometrical point of view, we view $\bbF(x)$ and $\bbF(y)$ as the field of
(rational) functions on $\bbF^n$, respectively on $\bbF^{|\wght|}$; the Poisson structure on $\bbF^n$ and
on~$\bbF^{|\wght|}$ corresponding to $\PB_b$ and $\PB_b^\wght$ will respectively be denoted by $\pi_b$ and
$\pi_b^\wght$.  The standard Lotka-Volterra Hamiltonians on $\bbF^n$ and on $\bbF^{|\wght|}$, which are always the
sum of all coordinates, are denoted by $H$ and $H^\wght$.

In order to show the integrability of the Hamiltonian system $(\bbF^{|\wght|},\pi_b^\wght,H^\wght)$ and of its
Kahan discretization, we first recall from \cite{PPP_Red_DC} the integrability of the deformed Lotka-Volterra
system $(\bbF^n,\pi_b,H)$ and its Kahan discretization. For $1\leqs\ell\leqs \left[\frac{n-1}2\right]$, consider
the following rational functions:
\begin{equation}\label{E:integrals}
  F_\ell:= \left\{
  \begin{array}{ll}
    \left(x_1+x_2+\cdots+x_{2\ell-1}+\frac{\beta}{x_n}\right)
    \displaystyle\frac{x_{2\ell+1}x_{2\ell+3}\ldots x_{n}}{x_{2\ell}x_{2\ell+2}\ldots x_{n-1}}\;,&
    \mbox{ if}\  n\mbox{ is odd},\\
    \\
    \left(x_1+x_2+\cdots+x_{2\ell}+\frac{\beta}{x_n}\right)
    \displaystyle\frac{x_{2\ell+2}x_{2\ell+4}\ldots x_n}{x_{2\ell+1} x_{2\ell+3}\ldots x_{n-1}}\;,&
   \mbox{ if} \  n\mbox{ is even},
  \end{array} \right.
\end{equation}
and let $G_\ell:=\jmath(F_\ell)$, where $\jmath:\bbF(x)\to\bbF(x)$ is the involutive field automorphism, defined by
$\jmath(x_i):=x_{n+1-i}$, for $i=1,\dots,n$. Together with the Hamiltonian $H$, this yields exactly $n-1$ different
rational functions: for example, when~$n$ is odd then all $F_\ell$ and $G_\ell$ are different, except for
$F_1=G_1$.  The following facts were obtained in \cite{PPP_Red_DC}:
\begin{enumerate}
  \item The $n-1$ rational functions $F_\ell$, $G_\ell$ and $H$ are first integrals of $(\bbF^n,\pi_b,H)$;
  \item They are \emph{independent}, i.e., their differentials are independent on an open dense subset of $\bbF^n$;
  \item The rank $\Rk\pi_b$ of the Poisson structure $\pi_b$ is $n$ when $n$ is even, otherwise it is $n-1$;
  \item The first integrals $F_\ell$ are \emph{in involution}, i.e., commute for the Poisson bracket;
  \item The first integrals $F_\ell, G_\ell$ and $H$ are invariants of the Kahan discretization of $(\bbF^n,\pi_b,H)$.
\end{enumerate}
Items (1) and (2) say that $(\bbF^n,\pi_b,H)$ is \emph{superintegrable}, i.e., has $n-1$ independent first
integrals, where $n$ is the dimension of the ambient space. The items (1) -- (4) imply that the independent first
integrals $F_\ell$ are in involution and that, with the Hamiltonian, their number is $n-\frac12\Rk\pi_b$, which is
exactly the number required for the \emph{Liouville integrability} of $(\bbF^n,\pi_b,H)$; for example, when~$n$ is
odd, $\Rk\pi_b=n-1$ and we have $(n+1)/2$ functions $F_1,\dots,F_{(n-1)/2},H$ which are in involution. Combined
with (5) and the fact that the Kahan map is a Poisson map one gets from it that the Kahan discretization of
$(\bbF^n,\pi_b,H)$ is both superintegrable and Liouville integrable, with as invariants the first integrals of the
continuous system.

We use these five properties, the properties of the decloning map and of the Poisson structure~$\pi_b^\wght$ to
prove the integrability the Hamiltonian system $(\bbF^{|\wght|},\pi_b^\wght,H^\wght)$ and its Kahan
discretization. As we have already seen in the proof of Proposition \ref{prp:kahan_cloning_defo}, for any
$1\leqslant i\leqs n$ and $1<k\leqslant\wght(i)$, $y_{i,k}/y_{i,1}$ is a Casimir function of $\pi_b^\wght$,
yielding $|\wght|-n$ different Casimir functions, which are clearly independent. The rank of $\pi_b^\wght$ is
therefore at most $n$, i.e., at most $n-1$ when~$n$ is odd, and at most $n$ when $n$ is even. In fact, we have
equality. To see this, consider the \emph{decloning map} $S:\bbF^{|\wght|}\to\bbF^n$ corresponding to decloning
morphism $\Sigma$, i.e., $\Sigma=S^*$, which is a dominant (actually surjective) Poisson map, since $\Sigma$ is an
injective Poisson morphism. It follows that the rank of $\pi_b^\wght$ is also bounded from below by the rank of
$\pi_b$, which is $n-1$ when $n$ is odd, and $n$ when $n$ is even, and so we have equality.

The cited properties of the decloning map $S$ also imply that the $n-1$ functions $S^*(F_\ell)$, $S^*(G_\ell)$ and
$H^\wght$ are independent first integrals of $(\bbF^{|\wght|},\pi_b^\wght,H^\wght)$ and that they are in
involution. Combined with the Casimir functions we get $n-1+|\wght|-n=|\wght|-1$ different functions; from the
simple form of the Casimir functions, it is clear that the former first integrals are independent from the Casimir
functions, so that we have $|\wght|-1=\dim\bbF^{|\wght|}-1$ independent first integrals, which shows that
$(\bbF^{|\wght|},\pi_b^\wght,H^\wght)$ is superintegrable.

Since $S$ is a Poisson map, the independent functions $S^*(F_\ell)$ and $H^\wght=S^*(H)$ are in involution;
together with the $|\wght|-n$ Casimir functions, we get $(n+1)/2+|\wght|-n= |\wght|-\frac12\Rk\pi_b^\wght$
independent functions, including the Hamiltonian, in involution, which shows that
$(\bbF^{|\wght|},\pi_b^\wght,H^\wght)$ is Liouville integrable.

Moreover, the $|\wght|-1$ independent first integrals are invariants of $\Kb^\wght$. For the Casimir functions
$y_{i,k}/y_{i,1}$, we have already seen this in the proof of Proposition \ref{prp:kahan_cloning_defo}. For
$1\leqs\ell\leqs \left[\frac{n-1}2\right]$, the commutativity of \eqref{dia:kahan_diagram_2} implies that
$$
  \Kb^\wght(\Sigma(F_\ell))=\Sigma(\Kb(F_\ell)) = \Sigma(F_\ell)\;,
$$
where we have used in the last step that $F_\ell$ is an invariant of $\Kb$, $\Kb(F_\ell)=F_\ell$. Since $H^\wght$ is
linear, it is also an invariant of $\Kb^\wght$.

Summing up, and combined with the integrability results in the non-deformed case, it leads to the following
proposition.

\begin{prop}
  Let $n\in\bbN^*$ and let $\wght$ be any weight vector on $\Ga n$. Suppose that $\Delta_n^\wght$ is any augmented
  graph of $\Ga n^\wght$.
  \begin{enumerate}
    \item The Lotka-Volterra system $\LV(\Ga n^\wght)$ and its Kahan discretization are superintegrable and
      Liouville integrable;
    \item The deformed Lotka-Volterra system $\LV(\Delta_n^\wght)$ and its Kahan discretization are superintegrable
      and Liouville integrable.
  \end{enumerate}\qed
\end{prop}

The Lotka-Volterra and deformed Lotka-Volterra systems having Kahan discretizations which are integrable with
respect to the original Poisson structure are therefore characterized by the \KP.
%
%

\bibliographystyle{abbrv} 

\end{document}